\documentclass[a4paper]{article}
\usepackage{geometry}
\usepackage{amsmath,amsfonts,amssymb,amsthm,stmaryrd}
\usepackage[small]{diagrams}
\usepackage{hyperref}
\usepackage{bm}

\usepackage{sectsty}
\sectionfont{\sffamily}
\subsectionfont{\sffamily}
\subsubsectionfont{\sffamily}

\newtheoremstyle{myplain}
  {\topsep}   
  {\topsep}   
  {\itshape}  
  {0pt}       
  {\bfseries\sffamily} 
  {.}         
  {5pt plus 1pt minus 1pt} 
  {}          
  
\newtheoremstyle{mydefinition}
  {\topsep}   
  {\topsep}   
  {\normalfont}  
  {0pt}       
  {\bfseries\sffamily} 
  {.}         
  {5pt plus 1pt minus 1pt} 
  {}          
  
\newtheoremstyle{myremark}
  {0.5\topsep}   
  {0.5\topsep}   
  {\normalfont}  
  {0pt}       
  {\sffamily} 
  {.}         
  {5pt plus 1pt minus 1pt} 
  {}          

\theoremstyle{mydefinition}
\newtheorem{definition}{$\blacktriangleright$ Definition}[subsection]
\theoremstyle{myplain}
\newtheorem{theorem}[definition]{$\blacktriangleright$ Theorem}
\newtheorem{lemma}[definition]{$\blacktriangleright$ Lemma}

\newtheorem{proposition}[definition]{$\blacktriangleright$ Proposition}
\theoremstyle{myremark}
\newtheorem*{remark}{$\blacktriangleright$ Remark}
\newtheorem*{notation}{$\blacktriangleright$ Notation}

\newtheorem*{convention}{$\blacktriangleright$ Convention}

\newarrow{Corresponds}<--->

\newcommand{\prd}{\hfill$\blacktriangleleft$}

\begin{document}
\title{\Huge \bf Game-theoretic Interpretation of Type Theory Part II \\ \Large Uniqueness of Identity Proofs and Univalence}
\author{\Large \bf Norihiro Yamada \\ \\
{\tt norihiro.yamada@cs.ox.ac.uk} \\
Department of Computer Science \\
University of Oxford
}

\maketitle

\begin{abstract}
In the present paper, based on the previous work (Part I), we present a game semantics for the intensional variant of intuitionistic type theory that refutes the principle of uniqueness of identity proofs and validates the univalence axiom, though we do not interpret non-trivial higher propositional equalities. Specifically, following the historic groupoid interpretation by Hofmann and Streicher, we equip predicative games in Part I with a \emph{groupoid} structure, which gives rise to the notion of \emph{(predicative) gamoids}. Roughly, gamoids are \emph{``games with (computational) equalities specified''}, which interpret subtleties in Id-types. We then formulate a category with families of predicative gamoids, equipped with $\prod$-, $\sum$- and Id-types as well as universes, which forms a concrete instance of the groupoid model. We believe that this work is an important stepping-stone towards a complete interpretation of homotopy type theory.
\end{abstract}
\tableofcontents

\section{Introduction}
This paper, a continuation of Part I (\cite{yamada2016game}), presents a game semantics for the intensional variant of intuitionistic type theory (ITT) \cite{martin1984intuitionistic, martin1998intuitionistic, rose1984intuitionistic} that refutes the principle of uniqueness of identity proofs (UIP), and admits the univalence axiom (UA) as well as the axiom of function extensionality (FunExt), though we do not interpret non-trivial higher propositional equalities. 

In the previous work \cite{yamada2016game}, we presented an interpretation of ITT with $\prod$-, $\sum$-, and Id-types as well as the hierarchy of universes in terms of games and strategies. Specifically, we gave an instance of a category with families (CwF) \cite{dybjer1996internal, hofmann1997syntax} (a categorical model of dependent type theory) $\mathcal{IPG}$ of \emph{predicative games} and \emph{generalized strategies}, a generalization of games and strategies in the existing game semantics \cite{abramsky2000full,hyland2000full,mccusker1998games} that achieves an interpretation of dependent types and universes in a systematic way. However, it admits UIP and refutes UA; thus, it does not completely interpret phenomena of propositional equalities in homotopy type theory (HoTT) \cite{voevodsky2013homotopy}, a recent variant of ITT based on the homotopy-theoretic interpretation, which gets much attention from the community of mathematics, logic, and computer science.

Meanwhile, we recognized that $\mathcal{IPG}$ appears quite similar to the CwF $\mathcal{GPD}$ of groupoids in the historic paper \cite{hofmann1998groupoid}, which showed for the first time that UIP is not derivable in ITT. Also, it later led to the homotopy-theoretic interpretation of the type theory, which resulted in HoTT.  From this observation, we equip $\mathcal{IPG}$ with a groupoid structure, forming the CwF $\mathcal{PGD}$ of ``\emph{predicative gamoids}''. Roughly, predicative gamoids are \emph{``predicative games with (computational) equalities specified''}. This additional structure refines the interpretation of propositional equalities, which is the main improvement from the previous work \cite{yamada2016game}. 
Importantly, $\mathcal{PGD}$ has turned out to be a concrete instance (more precisely a ``subcategory with families'') of $\mathcal{GPD}$, though we did not particularly intend to do so.  

As a result, our model in $\mathcal{PGD}$ interprets phenomena in HoTT better than the model in $\mathcal{IPG}$: It refutes UIP, and admits UA as well as FunExt. However, it still fails to interpret the infinite hierarchy of Id-types.
Thus, as a future work, we shall generalize predicative gamoids to form an instance of \emph{$\omega$-groupoids} in order to capture non-trivial higher equalities (as several papers such as \cite{warren2011strict, van2011types, lumsdaine2009weak} did, though their models are abstract, categorical ones). Also, it remains to obtain definability and full abstraction results.

The rest of the paper is structured as follows. We define a category of predicative gamoids in Section \ref{PredicativeGamoids}, and give some constructions on them in Section \ref{ConstructionsOnGamoids}. We then define a CwF of predicative gamoids in Section \ref{GTIITT}, which is a highlight of the present paper. We finally investigate some of its properties in Section~\ref{Intensionality}; in particular, we show that it refutes UIP and admits UA.

\begin{remark}
Throughout the present paper, we assume that the reader is familiar with the basic notions and results in Part I (\cite{yamada2016game}).
\end{remark}

We ends this introduction by introducing a convenient notation:
\begin{notation}
Let $A, B$ be predicative games, and $\sigma : A$, $\tau : B$ strategies. We write $\sigma \! \gg \! \tau$ for the strategy on the linear implication $A \! \multimap \! B$ that plays as $\tau$ up to the ``tags for disjoint union''. Moreover, if $\phi (\sigma_1, \dots, \sigma_n) : B$ is a composition of strategies $\sigma_1 : A_1, \dots, \sigma_n : A_n$ and possibly some other strategies, then $(\sigma_1 \otimes \dots \otimes \sigma_n) \! \rightleftharpoons \! \phi (\sigma_1, \dots, \sigma_n)$ denotes the strategy on $A \! \multimap \! B$, where $A \stackrel{\mathrm{df. }}{=} A_1 \otimes \dots \otimes A_n$ that plays as $\phi (\sigma_1, \dots, \sigma_n)$ plus ``copy-cat'' between the two occurrences of $\sigma_i$ for all $i = 1, \dots, n$. Abusing the notation, we apply these notations for the implication $A \to B = \ !A \! \multimap \! B$ as well.
\end{notation}

\pagebreak
\section{Predicative Gamoids}
\label{PredicativeGamoids}
We begin with equipping \emph{predicative games} defined in the previous paper \cite{yamada2016game} with a \emph{groupoid} structure. We shall call the resulting notion \emph{predicative gamoids}, which can be considered as \emph{``predicative games with (computational) equalities specified''}. Note that this construction is applicable for other variants of games as well; such a more general notion should be called \emph{gamoids}.

\subsection{Gamoids}
We first define the general notion of \emph{gamoids}. The underlying games and strategies here may be arbitrary; but for concreteness, ``games'' and ``strategies'' in this section refer to the variant in \cite{mccusker1998games} (as in the previous paper \cite{yamada2016game}).

Let $G$ be a game, and $\sigma, \sigma' : G$ strategies. Recall that \emph{quasi-copy-cat strategies} (or \emph{qcc strategies} for short) $p : \sigma \simeq_G \! \sigma'$ are strategies $p : \sigma \! \multimap \! \sigma'$ that are history-free isomorphisms respecting labels and justifiers (for the precise definition, see \cite{yamada2016game}). It has been shown in \cite{yamada2016game} that qcc strategies $p : \sigma \simeq_G \! \sigma'$ are a kind of graph isomorphisms between the strategies $\sigma$ and $\sigma'$. Thus, qcc strategies appear an appropriate notion of ``isomorphisms'' between strategies, but to interpret the univalence axiom, we need to relax them to isomorphism strategies (see Section~\ref{Intensionality}). Also, recall that a \emph{groupoid} is a category whose morphisms are all isomorphisms (i.e., invertible morphisms).

We now define the notion of \emph{gamoids}:
\begin{definition}[Gamoids]
A {\bf gamoid} is a groupoid such that:
\begin{itemize}

\item Objects are strategies on a fixed game.

\item Morphisms are isomorphism strategies.

\end{itemize}
Morphisms in a gamoid are called {\bf identification strategies} (or {\bf identifications} for short).
\prd
\end{definition}

As the name suggests, morphisms in a gamoid are intended to be \emph{``computational witnesses of equalities''} between strategies on the underlying game. They are ``computational'' because strategies represent algorithms or constructive proofs.

\begin{notation}
A gamoid is usually specified by a pair $(G, =_G)$ of the underlying game $G$ and the set $=_G$ of its morphisms. We often abbreviate it as $G$. Moreover, we often write $\sigma =_G \! \sigma'$ for the hom-set $G(\sigma, \sigma')$, and $p : \sigma =_G \! \sigma'$ when $p \in G(\sigma, \sigma')$. 
\end{notation}

Abusing the notation, we write $\sigma =_G \! \sigma'$ and say that  $\sigma$ and $\sigma'$ are {\bf (computationally) equal} if $\sigma =_G \! \sigma' \neq \emptyset$. We then need to distinguish two kinds of equalities:

\begin{convention}
We write $=$ without any subscript for equality ``on the nose'', called the {\bf strict equality}, between strategies, while we always put the subscript $G$ for the {\bf computational equality} $=_G$ in a gamoid $G$. Also, an {\bf equality} refers to a \emph{computational} equality by default.
\end{convention}

\begin{remark}
It is a very important point that an identification strategy $p : \sigma =_G \! \sigma'$ in a gamoid $G$ is an isomorphism strategy but \emph{not necessarily from $\sigma$ to $\sigma'$}. We need this arbitrariness because an ``equality'' in mathematics often focuses on some relevant or partial information of objects, e.g., consider the congruence relation on integers. This point will be illustrated later by concrete examples.
\end{remark}

Conceptually, a gamoid $(G, =_G)$ is a game $G$ equipped with a set $=_G$ of (selected) isomorphism strategies that determines a computational equality between its strategies. It is a groupoid because an equality must be an equivalence relation (i.e., a reflexive, symmetric, and transitive relation). We emphasize here that we derived this concept directly from the seminal groupoid interpretation by Hofmann and Streicher \cite{hofmann1998groupoid}. Intuitively, computational equalities should be weaker than strict equalities as ``what can be observed to be the same'' may be limited, which is achieved by our definition above. More importantly, the notion of gamoids enables us to ``tailor'' an equality in each game. E.g., in the usual game semantics such as \cite{abramsky2000full,hyland2000full,mccusker1998games}, we have $\underline{n} \simeq_{\mathcal{N}} \underline{m}$, where $\mathcal{N}$ is the \emph{natural numbers game}, and $\underline{n}, \underline{m} : \mathcal{N}$ are strategies for natural numbers $n, m$, respectively, even if $n \neq m$; thus, the gamoid $\mathcal{N}$ of natural numbers must not count this isomorphism strategy as an identification.

In mathematics, this phenomenon is everywhere: We have some ``space'' of objects equipped with equalities between them (it may be more appropriate to say \emph{equivalences} or \emph{congruences}). Note that once such equalities between ``atomic objects'' have been defined, equalities between ``compound objects'' constructed from such atomic objects should be automatically determined. In other words, a construction on such spaces must operate on equalities as well. For example, pairs $(x, y)$, $(z, w)$ are equal iff $x, z$ are equal and $y, w$ are equal. We shall reflect these facts in the constructions on (predicative) gamoids in later sections.

Now, we introduce two particular kinds of gamoids: \emph{canonical} and \emph{discrete} gamoids:
\begin{proposition}[Games as groupoids]
Any game $G$ induces the {\bf canonical gamoid} $\mathcal{C}(G) = (G, \simeq_G)$ and the {\bf discrete gamoid} $\mathcal{D}(G) = (G, =)$, where $\simeq_G$ (resp. $=$) denotes the set of all isomorphism strategies (resp. strict equalities) between strategies on $G$.
\end{proposition}
\begin{proof}
Immediate from the definition.
\end{proof}
That is, a canonical gamoid $\mathcal{C}(G)$ is a game $G$ that equates all isomorphic strategies, while a discrete gamoid $\mathcal{D}(G)$ is a game $G$ that equates only strictly equal strategies. 

Next, the notion of \emph{subgamoids} is defined in the obvious way:
\begin{definition}[Subgamoids]
A {\bf subgamoid} of a gamoid $G$ is a gamoid $G'$ that is a subcategory of $G$.
In this case, we write $G' \leqslant G$. 
\prd
\end{definition}

\subsection{The Category of Predicative Gamoids}
From now on, we shall focus on a particular kind of gamoids, called \emph{predicative gamoids}:
\begin{definition}[Predicative gamoids]
A gamoid whose underlying structure is the category $\mathcal{PG}$ of predicative games and generalized strategies defined in \cite{yamada2016game} is called a {\bf predicative gamoid}.
\prd
\end{definition}

\begin{remark}
Strictly speaking, the above definition is not precise because the notion of gamoids is defined on the variant of games and strategies of \cite{mccusker1998games}. We meant that a predicative gamoid is the resulting structure when we apply the construction of gamoids to a predicative game.
\end{remark}

We proceed to define the category of predicative gamoids, which will form the ``universe'' to interpret intuitionistic type theory. But we first need a preliminary notion:
\begin{definition}[Equality-preserving strategies]
Let $A, B$ be gamoids. A strategy $\tau : A \! \multimap \! B$ is said to be {\bf equality-preserving} if it is equipped with an equality $\tau_p : \tau \circ \sigma_1 =_B \! \tau \circ \sigma_2$ in $B$ for each triple $(\sigma_1, \sigma_2, p)$ of strategies $\sigma_1, \sigma_2 : A$ and an equality $p : \sigma_1 =_A \! \sigma_2$ in $A$ such that the maps
\begin{align*}
(\sigma : A) &\mapsto \tau \circ \sigma \\
(p : \sigma_1 =_A \! \sigma_2) &\mapsto \tau_p : \tau \circ \sigma_1 =_B \! \tau \circ \sigma_2
\end{align*}
form a functor $A \to B$.
\prd
\end{definition}

Conceptually, the equality-preserving property is a reasonable requirement for strategies $\tau$ of (linear) function type: If inputs $\sigma_1, \sigma_2 : A$ are equal, then the outputs $\tau \circ \sigma_1, \tau \circ \sigma_2 : B$ must be equal. Also, from the category-theoretic point of view, it is a very natural definition: A gamoid is a game equipped with the structure of a category (groupoid), and an equality-preserving strategy is a strategy equipped with the structure of a functor. From a different angle, we may say that an equality-preserving strategy plays ``in two dimensions'' as well in addition to the usual play ``in one dimension'' of a mere strategy. Again, we note that this formulation coincides with that of the groupoid interpretation in \cite{hofmann1998groupoid}.

Before proceeding further, let us recall some notations from the previous paper:
\begin{notation}
We write $\&$ and $!$ for paring and promotion of generalized strategies, respectively.
Also, we usually write $A \to B$ for the linear implication $!A \! \multimap \! B$. Moreover, for generalized strategies $\phi : A \to B$, $\psi : B \to C$, we define $\psi \bullet \phi \stackrel{\mathrm{df. }}{=} \psi \ \! \circ \ \! ! \phi : A \to C$. 
\end{notation}

We now define the category of predicative gamoids:
\begin{definition}[The category $\mathcal{PGD}$]
\label{PGD}
The category $\mathcal{PGD}$ of predicative gamoids and equality-preserving (generalized) strategies is defined as follows:
\begin{itemize}

\item Objects are predicative gamoids.

\item A morphism $A \to B$ is an equality-preserving generalized strategy $\phi : A \to B$.

\item The composition of morphisms is the composition of functors, i.e., the composition of morphisms $\phi : A \to B$, $\psi : B \to C$ is the composition $\psi \bullet \phi : A \to C$ of strategies, equipped with the equality $(\psi \bullet \phi)_p \stackrel{\mathrm{df. }}{=} \psi_{\phi_p} : (\psi \bullet \phi) \bullet \sigma_1 =_C \! (\psi \bullet \phi) \bullet \sigma_2$ for each $\sigma_1, \sigma_2 : A$, $p : \sigma_1 =_A \! \sigma_2$.

\item The identity $\mathsf{id}_A$ on each object $A$ is the dereliction $\mathsf{der}_A$ with the equality $(\mathsf{der}_A)_p \stackrel{\mathrm{df. }}{=} p$ for each $\sigma_1, \sigma_2 : A$, $p : \sigma_1 =_A \! \sigma_2$.

\end{itemize}
\prd
\end{definition}

\begin{remark}
Strictly speaking, an equality-preserving generalized strategy $\phi : A \to B$ is equipped with an equality $\phi_p : \phi \circ \sigma_1 =_B \! \phi \circ \sigma_2$ for each $\sigma_1, \sigma_2 : \ ! A$, $p : \sigma_1 =_{!A} \! \sigma_2$. However, since there is an obvious bijection between such triples $\sigma_1, \sigma_2 : \ ! A$, $p : \sigma_1 =_{!A} \! \sigma_2$ in $!A$ and triples $\sigma'_1, \sigma'_2 : A$, $p' : \sigma'_1 =_{A} \! \sigma'_2$ in $A$, we shall focus on strategies and equalities in $A$ when we talk about inputs of such a strategy $\phi : A \to B$, and write $\phi_p$ rather than $\phi_{!p}$.
\end{remark}

Morphisms in $\mathcal{PGD}$ are strategies with the structure of functors, which differs from the category $\mathcal{IPG}$ in \cite{yamada2016game}. This means that such morphisms see the correspondence not only between strategies but also \emph{between identifications}, which will be the structure to interpret Id-types. 

Of course, we need to establish:
\begin{proposition}[Well-defined $\mathcal{PGD}$]
The structure $\mathcal{PGD}$ forms a well-defined category.
\end{proposition}
\begin{proof}
First, the composition of morphisms is well-defined and associative because it is just the composition of functors. 

Next, for each object $A$, the dereliction $\mathsf{der}_A$ with the equality $(\mathsf{der}_A)_p \stackrel{\mathrm{df. }}{=} p$ is clearly a well-defined morphism $A \to A$. The unit law for the object-map is obvious; for the arrow-map, observe that, for any morphism $\phi : A \to B$, we have
\begin{align*}
(\phi \bullet \mathsf{der}_A)_p &= \phi_{(\mathsf{der}_A)_p} = \phi_{p} \\
(\mathsf{der}_B \bullet \phi)_p &= (\mathsf{der}_B)_{\phi_p} = \phi_p
\end{align*}
for all $\sigma_1, \sigma_2 : A$, $p : \sigma_1 =_A \! \sigma_2$. Thus, $\phi \bullet \mathsf{der}_A = \phi = \mathsf{der}_B \bullet \phi$ as morphisms in $\mathcal{PGD}$.
\end{proof}

\pagebreak
\section{Constructions on Gamoids}
\label{ConstructionsOnGamoids}
In this section, we refine the constructions on predicative games in the previous paper \cite{yamada2016game} to define the corresponding constructions on predicative gamoids. This means that we shall equip them with operations on (computational) equalities.

Importantly, our aim is to capture phenomena in homotopy type theory (HoTT) in terms of games and strategies, in particular \emph{propositional equalities} by computational equalities. Thus, we shall define constructions on equalities reflecting properties of propositional equalities in HoTT, e.g., we will define equalities between dependent functions in such a way that the principle of function extensionality holds (see Section~\ref{Intensionality}).

\subsection{Dependent Gamoids}
Recall that, in \cite{yamada2016game}, we interpreted a dependent type by the notion of a \emph{dependent game}, which is basically a collection of games indexed over strategies on a fixed game. But now, we have a categorical structure on games, so we can refine this notion as \emph{functors}. 
\begin{definition}[Dependent gamoids]
\label{DependentGamoids}
A {\bf dependent gamoid} on a predicative gamoid $A$ is a functor $B : A \to \mathcal{PGD}$.
\prd
\end{definition}

\begin{remark}
Alternatively, we may have defined a dependent gamoid $B$ on a predicative gamoid $A$ as a morphism $B : A \to \mathcal{U}$ in the category $\mathcal{PGD}$, where $\mathcal{U}$ is an appropriate \emph{universe gamoid} (for this, we need to require that the identification $B_p : B \bullet \sigma_1 =_\mathcal{U} \! B \bullet \sigma_2$ is again an equality-preserving strategy for any $p : \sigma_1 =_A \! \sigma_2$). This in fact seems a reasonable definition, and we could take this route without any problem; however, for simplicity, we defined dependent gamoids as in Definition~\ref{DependentGamoids}.
\end{remark}

A dependent gamoid $B : A \to \mathcal{PGD}$ yields not only a predicative gamoid $B\sigma \in \mathcal{PGD}$ for each strategy $\sigma : A$ but also an isomorphism functor $Bp : B\sigma_1 \stackrel{\simeq}{\to} B\sigma_2$ for each equality $p : \sigma_1 =_A \! \sigma_2$.
Thus, e.g., we can automatically interpret Leibniz' law. 

Below, we proceed to define constructions on predicative gamoids, in which dependent gamoids are involved.

\subsection{Dependent Union}
Based on the constructions on predicative games defined in \cite{yamada2016game}, we shall define the corresponding constructions on predicative gamoids. We begin with the construction of \emph{dependent union}, which will serve as a convenient preliminary notion.

\begin{notation}
For readability, we write $\odot$ and $(\_)^\star$ for the composition and inverse in dependent unions, respectively.
\end{notation}

\begin{definition}[Dependent union]
Given a dependent gamoid $B : A \to \mathcal{PGD}$, its {\bf dependent union} $\uplus B$ is defined as follows:
\begin{itemize}
\item The underlying predicative game $\uplus B$ is the subgame of the dependent union of the underlying dependent game $B$ over the underlying predicative game $A$ whose strategies are all equality-preserving.

\item For any objects $\tau, \tau' : \uplus B$, we define 
\begin{equation*}
\tau =_{\uplus B} \! \tau' \stackrel{\mathrm{df. }}{=} \textstyle \bigcup_{p : \sigma =_A \sigma'} Bp \bullet \tau =_{B\sigma'} \! \tau' 
\end{equation*}
where $\sigma, \sigma' : A$, $\tau : B\sigma$, $\tau' : B\sigma'$.

\item The composition $q' \! \odot q$ of identifications $q : \tau =_{\uplus B} \! \tau'$, $q' : \tau' =_{\uplus B} \! \tau''$, where $\tau : B\sigma$, $\tau' : B\sigma'$, $\tau'' : B\sigma''$, $q : Bp \bullet \tau =_{B\sigma'} \! \tau'$, $q' : Bp' \bullet \tau' =_{B\sigma''} \! \tau''$ for some $\sigma, \sigma', \sigma'' : A$, $p : \sigma =_A \! \sigma'$, $p' : \sigma' =_A \! \sigma''$, is defined by
\begin{equation*}
q' \! \odot q \stackrel{\mathrm{df. }}{=} q' \! \circ (Bp')_q : B(p' \! \circ p) \bullet \tau =_{B\sigma''} \! \tau''
\end{equation*}
where $\sigma \! \stackrel{p}{=}_A \! \sigma' \! \stackrel{p'}{=}_A \! \sigma''$.

\item The identity on each object $\tau : B\sigma$ is the identity $\mathsf{id}_\tau$ in $B\sigma$. 

\end{itemize}
\prd
\end{definition}

\begin{remark}
Note that identifications $q : \tau =_{\uplus B} \! \tau'$, where $\tau : B\sigma$, $\tau' : B\sigma'$, $\sigma, \sigma' : A$, in a dependent union $\uplus B$ are strategies not between $\tau$ and $\tau'$ but between $Bp \bullet \tau$ and $\tau'$ for some $p : \sigma =_A \! \sigma'$. This is one of the main reasons why identifications $\zeta =_G \! \zeta'$ in a gamoid $G$ in general are not necessarily between the strategies $\zeta$ and $\zeta'$ themselves.
\end{remark}

Of course, we need to establish the following:
\begin{proposition}[Well-defined dependent union]
For any dependent gamoid $B : A \to \mathcal{PGD}$, the dependent union $\uplus B$ is a well-defined predicative gamoid.
\end{proposition}
\begin{proof}
It is straightforward to see that the composition and identities are well-defined, where the functoriality of $B$ is essential. The inverse $q^\star$ of each identification $q : Bp \bullet \tau =_{B\sigma'} \! \tau'$ is given by:
\begin{equation*}
q^\star \stackrel{\mathrm{df. }}{=} (Bp^{-1})_{q^{-1}} : Bp^{-1} \! \bullet \tau' =_{B\sigma} \! \tau
\end{equation*}
where $\sigma, \sigma' : A$, $\tau : B\sigma$, $\tau' : B\sigma'$, $p : \sigma =_A \! \sigma'$.
In fact, $q^\star \! \odot q = (Bp^{-1})_{q^{-1}} \! \circ (Bp^{-1})_q = (Bp^{-1})_{q^{-1} \circ q} = (Bp^{-1})_{\mathsf{id}_{Bp \bullet \tau}} = \mathsf{id}_{Bp^{-1} \bullet Bp \bullet \tau} = \mathsf{id}_{B(p^{-1} \circ p) \bullet \tau} = \mathsf{id}_{B(\mathsf{id}_\sigma) \bullet \tau} = \mathsf{id}_{\mathsf{id}_{B\sigma} \bullet \tau} = \mathsf{id}_\tau$, and
$q \odot q^\star = q \circ (Bp)_{(Bp^{-1})_{q^{-1}}} = q \circ (Bp \bullet Bp^{-1})_{q^{-1}} = q \circ B(p \circ p^{-1})_{q^{-1}} = q \circ B(\mathsf{id}_{\sigma'})_{q^{-1}} = q \circ (\mathsf{id}_{B\sigma'})_{q^{-1}} = q \circ q^{-1} = \mathsf{id}_{\tau'}$.

For associativity of the composition, let $q : \tau =_{\uplus B} \! \tau'$, $q' : \tau' =_{\uplus B} \! \tau''$, $q'' : \tau'' =_{\uplus B} \! \tau'''$, where $\tau : B\sigma$, $\tau' : B\sigma'$, $\tau'' : B\sigma''$, $\tau''' : B\sigma'''$ and $q : Bp \bullet \tau =_{B\sigma'} \! \tau'$, $q' : Bp' \bullet \tau' =_{B\sigma''} \! \tau''$, $q'' : Bp'' \bullet \tau'' =_{B\sigma'''} \! \tau'''$ for some $\sigma, \sigma', \sigma'', \sigma''' : A$, $p : \sigma =_A \! \sigma'$, $p' : \sigma' =_A \! \sigma''$, $p'' : \sigma'' =_A \! \sigma'''$. Then observe that:
\begin{align*}
q'' \! \odot (q' \! \odot q) &= q'' \! \odot (q' \circ (Bp')_q) \\
&= q'' \circ (Bp'')_{q' \circ (Bp')_q} \\
&= q'' \circ ((Bp'')_{q'} \circ (Bp'')_{(Bp')_q}) \\
&= (q'' \circ (Bp'')_{q'}) \circ (Bp'')_{(Bp')_q} \\
&= (q'' \! \odot q') \circ (Bp'' \! \bullet Bp')_q \\
&= (q'' \! \odot q') \circ B(p'' \! \circ p')_q \\
&= (q'' \! \odot q') \odot q.
\end{align*}

For the unit law, observe that:
\begin{align*}
q \odot \mathsf{id}_\tau &= q \circ (Bp)_{\mathsf{id}_\tau} = q \circ \mathsf{id}_{Bp \bullet \tau} = q \\
\mathsf{id}_{\tau'} \! \odot q &= \mathsf{id}_{\tau'} \! \circ B(\mathsf{id}_{\sigma'})_q = \mathsf{id}_{\tau'} \! \circ (\mathsf{id}_{B\sigma'})_q = \mathsf{id}_{\tau'} \! \circ q = q
\end{align*}
which completes the proof.
\end{proof}

\subsection{Dependent Function Space}
Next, we consider the construction of \emph{dependent function space} that is intended to interpret dependent function types. 
We already defined dependent function spaces $\widehat{\prod}(A, B)$ of dependent games $B : A \to \mathcal{PG}$ in the previous paper \cite{yamada2016game}; the challenge here is how to define \emph{identifications in dependent function games}. As mentioned earlier, we shall define them reflecting phenomena in HoTT; in the case of dependent functions, we take the ``point-wise'' identification as the definition: An identification $q : \tau =_{\widehat{\prod}(A, B)} \! \tau'$ is defined to be a family 
\begin{equation*}
(q_\sigma : \tau \bullet \sigma =_{B\sigma} \! \tau' \bullet \sigma)_{\sigma : A}
\end{equation*}
of identifications in $\uplus B$.
Note that it coincides with the groupoid interpretation \cite{hofmann1998groupoid}. The remaining point is how to regard such a family as a \emph{single} isomorphism strategy; however, the structure of predicative games enables us to define $q \stackrel{\mathrm{df. }}{=} \& \{ q_\sigma \ \! | \ \! \sigma : A \}$.

\begin{definition}[Dependent function space]
Given a dependent gamoid $B : A \to \mathcal{PGD}$, the {\bf dependent function space} $\widehat{\prod}(A, B)$ of $B$ over $A$ is defined as follows:
\begin{itemize}

\item The underlying predicative game $\widehat{\prod}(A, B)$ is the subgame of the dependent function space of the underlying dependent game $B$ over the underlying predicative game $A$ whose strategies $\tau$ are equality-preserving and satisfy $\tau_p : Bp \bullet (\tau \bullet \sigma_1) =_{B\sigma_2} \! \tau \bullet \sigma_2$ for all $\sigma_1, \sigma_2 : A, p : \sigma_1 =_A \sigma_2$. Objects of a dependent function space are called {\bf dependent functions}.

\item For each pair $\tau, \tau' : \widehat{\prod}(A, B)$ of objects, the hom-set $\tau =_{\widehat{\prod}(A, B)} \! \tau'$ is defined by
\begin{equation*}
\tau =_{\widehat{\prod}(A, B)} \! \tau' \stackrel{\mathrm{df. }}{=} \{ \& \{ q_\sigma \ \! | \! \ \sigma : A \} \ \! | \! \ \text{$q : \tau \to \tau'$ is a natural transformation} \}.
\end{equation*}
We usually write $q$ for an identification $\& \{ q_\sigma \ \! | \! \ \sigma : A \}$.

\item The composition of identifications $q: \tau =_{\widehat{\prod}(A,B)} \! \tau'$, $q': \tau' =_{\widehat{\prod}(A,B)} \! \tau''$ is defined by 
\begin{equation*}
q' \! \circ q \stackrel{\mathrm{df. }}{=} \& \{ q'_\sigma \! \circ q_\sigma \ \! | \ \! \sigma : A \} : \tau =_{\widehat{\prod}(A,B)} \! \tau''
\end{equation*}
i.e., $(q' \! \circ q)_\sigma \stackrel{\mathrm{df. }}{=} q'_\sigma \! \circ q_\sigma$ for all $\sigma : A$. 

\item The identity $\textsf{id}_\tau$ on each object $\tau : \widehat{\prod}(A,B)$ is defined to be $\& \{ \mathsf{id}_{\tau \bullet \sigma} \ \! | \ \! \sigma : A \}$.

\end{itemize}
In particular, if $B$ is a ``constant dependent gamoid'' (i.e., a constant functor), then we write $A \to B$ for $\widehat{\prod}(A,B)$, and call it the {\bf implication} from $A$ to $B$.
\prd
\end{definition}

As one may immediately recognize, the dependent function space $\widehat{\prod}(A, B)$ is intended to be a generalization of the implication $A \to B$, in which the ambient game of outputs may depend on inputs. 

\begin{remark}
We require the naturality condition on identifications in dependent function spaces mainly in order to equip the \emph{evaluation strategy} $\mathsf{ev}$ (see \cite{abramsky2000full, mccusker1998games} for the definition) with the structure of a functor; see Lemma \ref{EV} below. 
\end{remark}

It is easy to see the following:
\begin{proposition}[Well-defined dependent function space]
For any dependent gamoid $B : A \to \mathcal{PGD}$, the dependent function space $\widehat{\prod}(A, B)$ is a well-defined predicative gamoid.
\end{proposition}
\begin{proof}
Since we already established that dependent unions are well-defined, it is easy to see that the dependent function space $\widehat{\prod}(A, B)$ is a well-defined predicative gamoid, in which the inverse of an identification $q : \tau =_{\widehat{\prod}(A,B)} \! \tau'$ is given by $q^{-1} = \& \{ q_\sigma^{-1} \ \! | \ \! \sigma : A \}$.
\end{proof}

Intuitively, if dependent functions $\tau, \tau' : \widehat{\prod}(A, B)$ are equal, then their outputs should be equal when they are applied to equal strategies on $A$. However, there is a difficulty to overcome: When $\sigma, \sigma' : A$ and $p : \sigma =_A \! \sigma'$, we have $\tau \bullet \sigma : B\sigma$ and $\tau' \bullet \sigma' : B\sigma'$, so they may be on different games; if it is the case, we cannot even talk about their equality. However, note that there is an isomorphism functor $Bp : B\sigma \stackrel{\simeq}{\to} B\sigma'$; we then utilize its object-map to ``transport'' $\tau \bullet \sigma : B\sigma$ into $B\sigma'$, resolving the ``type-unmatch'' problem. This is the idea described in \cite{hofmann1998groupoid, voevodsky2013homotopy}, on which our definition above is based.

Now, let us see that we have achieved what is described above:
\begin{proposition}[Equality-preservation]
\label{EqualityPreservation}
Let $B : A \to \mathcal{PGD}$ be a dependent gamoid, and $\tau, \tau' \in \widehat{\prod}(A, B)$. If $\tau =_{\widehat{\prod}(A,B)} \! \tau'$, then $Bp \bullet (\tau \bullet \sigma) =_{B\sigma'} \! \tau' \bullet \sigma'$ for any strategies $\sigma, \sigma' : A$ with an identification $p : \sigma =_A \sigma'$.
\end{proposition}
\begin{proof}
Assume that we have an identification $q : \tau =_{\widehat{\prod}(A,B)} \! \tau'$. Then we have at least two identifications $\tau'_p \odot q_\sigma, q_{\sigma'} \! \odot \tau_p : Bp \bullet (\tau \bullet \sigma) =_{B\sigma'} \! \tau' \bullet \sigma'$.
\end{proof}
Thus, as mentioned before, equal functions produce equal outputs when they take equal inputs. Note that, by naturality of $q$, the two identifications $\tau'_p \odot q_\sigma, q_{\sigma'} \! \odot \tau_p$ are actually the same, for both of which we write $q_p$.

\subsection{Dependent Pair Space}
We now proceed to define the construction of \emph{dependent pair space}, which is simpler than dependent function space as it is a generalization of product.
\begin{definition}[Dependent pair space]
Given a dependent gamoid $B : A \to \mathcal{PGD}$, the {\bf dependent pair space} $\widehat{\sum}(A, B)$ of $B$ over $A$ is defined as follows:
\begin{itemize}

\item The underlying predicative game $\widehat{\sum}(A, B)$ is the subgame of the dependent pair space of the underlying dependent game $B$ over the underlying predicative game $A$ whose strategies are all equality-preserving.

\item $\sigma \& \tau =_{\widehat{\sum}(A, B)} \! \sigma' \& \tau' \stackrel{\mathrm{df. }}{=} \{ p \& q \ \! | \ \! p : \sigma =_A \! \sigma', q : Bp \bullet \tau =_{B\sigma'} \! \tau' \}$ for all $\sigma, \sigma' : A$, $\tau : B\sigma$, $\tau' : B\sigma'$.

\item The composition of identifications $p \& q : \sigma \& \tau =_{\widehat{\sum}(A,B)} \! \sigma' \& \tau'$, $p' \& q' : \sigma' \& \tau'  =_{\widehat{\sum}(A,B)} \! \sigma'' \& \tau''$ is defined by:
\begin{equation*}
(p' \& q') \circ (p \& q) \stackrel{\mathrm{df. }}{=} (p' \! \circ p) \& (q' \! \odot q).
\end{equation*}

\item The identity $\mathsf{id}_{\sigma \& \tau}$ on each object $\sigma \& \tau$ is defined to be the paring $\mathsf{id}_\sigma \& \mathsf{id}_\tau$.

\end{itemize}
In particular, if $B$ is a ``constant dependent gamoid'', then we write $A \& B$ for the dependent pair space $\widehat{\sum}(A, B)$ and call it the {\bf product} of $A$ and $B$.
\prd
\end{definition}

As the name suggests, dependent pair spaces are intended to be a generalization of products, in which the ambient game of the second component can depend on the first component.

It is easy to see the following:
\begin{proposition}[Well-defined dependent pair space]
For any dependent gamoid $B : A \to \mathcal{PGD}$, the dependent pair space $\widehat{\sum}(A,B)$ is a well-defined predicative gamoid.
\end{proposition}
\begin{proof}
It is straightforward to see that the dependent pair space $\widehat{\sum}(A, B)$ is a well-defined category, e.g., the composition of identifications $p \& q : \sigma \& \tau =_{\widehat{\sum}(A,B)} \! \sigma' \& \tau'$, $p' \& q' : \sigma' \& \tau'  =_{\widehat{\sum}(A,B)} \! \sigma'' \& \tau''$, where $\sigma, \sigma', \sigma'' : A$, $\tau : B\sigma$, $\tau' : B\sigma'$, $\tau'' : B\sigma''$, is the product $(p' \! \circ p) \& (q' \! \odot q)$ with $p' \! \circ p : \sigma =_A \! \sigma''$ and $q' \! \odot q = q' \! \circ (Bp')_q : B(p' \! \circ p) \bullet \tau =_{B\sigma''} \! \tau''$, so it is an identification $\sigma \& \tau =_{\widehat{\sum}(A,B)} \! \sigma'' \& \tau''$, showing that the composition is well-defined.

It remains to show that each morphism is an isomorphism. As the inverse of $p \& q$, we take  $p^{-1} \& q^\star$, where $p^{-1} : \sigma' =_A \! \sigma$, $q^\star \! : Bp^{-1} \! \bullet \tau' =_{B\sigma} \! \tau$. Thus, $p^{-1} \& q^\star : \sigma' \& \tau' =_{\widehat{\sum}(A,B)} \! \sigma \& \tau$. Finally, it is straightforward to see that $(p \& q) \circ (p^{-1} \& q^\star) = \mathsf{id}_{\sigma' \& \tau'}$ and $(p^{-1} \& q^\star) \circ (p \& q) = \mathsf{id}_{\sigma \& \tau}$, showing that $p^{-1} \& q^\star = (p \& q)^{-1}$.
\end{proof}

\subsection{Id-gamoids}
Next, we define the construction of \emph{Id-gamoids}, which are to interpret Id-types.
\begin{definition}[Id-gamoids]
Given a predicative gamoid $G$ and objects $\sigma_1, \sigma_2 \in G$, the {\bf Id-gamoid} $\widehat{\textsf{Id}}_G(\sigma_1, \sigma_2)$ between $\sigma_1$ and $\sigma_2$ is defined to be the discrete gamoid $\mathcal{D}(\sigma_1 \! =_G \! \sigma_2)$, i.e., its objects are identifications between $\sigma_1$ and $\sigma_2$, equipped only with the trivial identifications between them.
\prd
\end{definition}

By the definition, Id-gamoids are clearly well-defined predicative gamoids.
Note that this definition follows the groupoid interpretation in \cite{hofmann1998groupoid}, which ``truncates'' all the non-trivial ``higher-morphisms''. To interpret the hierarchical structure of Id-types in intuitionistic type theory, we need to capture such non-trivial ``higher-identifications'', forming the structure of \emph{$\omega$-groupoids}; we shall address this problem in the next paper (Part III).

\pagebreak
\section{Game-theoretic Interpretation of ITT}
\label{GTIITT}
In this section, we present a category with families (CwF) of predicative gamoids, equipped with $\prod$-, $\sum$- and Id-types as well as universes. It can be seen as a refinement of the CwF $\mathcal{IPG}$ of predicative games in \cite{yamada2016game} by adding a groupoid structure. 

\subsection{Game-theoretic Category with Families for ITT}
For the definition of CwFs, see the standard references \cite{dybjer1996internal, hofmann1997syntax}, or the previous paper \cite{yamada2016game}. We now present our game-theoretic CwF of predicative gamoids.
\begin{definition}[The CwF $\mathcal{PGD}$]
We define the category with families $\mathcal{PGD}$ of predicative gamoids to be the structure $\mathcal{PGD} = (\mathcal{PGD}, \mathsf{Ty}, \mathsf{Tm}, \_\{\_\}, \mathcal{D}(I), \_.\_, \mathsf{p}, \mathsf{v}, \langle\_,\_\rangle_\_)$, where:
\begin{itemize}

\item The underlying category $\mathcal{PGD}$ is the category of predicative gamoids and equality-preserving strategies defined in Definition~\ref{PGD}. 

\item For each predicative gamoid $\Gamma \in \mathcal{PGD}$, we define $\mathsf{Ty}(\Gamma)$ to be the set of dependent gamoids on $\Gamma$.
In this situation, $\Gamma$ is particularly called a {\bf context gamoid}.

\item For a context gamoid $\Gamma \in \mathcal{PGD}$ and a dependent gamoid $A \in \mathsf{Ty}(\Gamma)$, we define $\mathsf{Tm}(\Gamma, A)$ to be the set of objects of the dependent function space $\widehat{\prod} (\Gamma, A)$.

\item For each morphism $\phi : \Delta \to \Gamma$ in $\mathcal{PGD}$, the function
\begin{equation*}
\_\{\phi\} : \mathsf{Ty}(\Gamma) \to \mathsf{Ty}(\Delta)
\end{equation*}
is defined by $A\{\phi\} \stackrel{\mathrm{df. }}{=} A \circ \phi : \Delta \to \mathcal{PGD}$ for all $A \in \mathsf{Ty}(\Gamma)$, and the functions
\begin{equation*}
\textstyle \_\{\phi\}_A : \mathsf{ob}(\widehat{\prod}(\Gamma, A)) \to \mathsf{ob}(\widehat{\prod}(\Delta, A\{\phi\}))
\end{equation*}
are defined by $\tau \{\phi\}_A \stackrel{\mathrm{df. }}{=} \tau \bullet \phi$ for all $A \in \mathsf{Ty}(\Gamma)$, $\tau \in \widehat{\prod}(\Gamma, A)$. In this context, $\phi$ is called a {\bf context morphism} from $\Delta$ to $\Gamma$, and the above functions are called {\bf substitutions} of $\phi$.

\item $\mathcal{D}(I)$ is the discrete gamoid, where $I$ the empty game $(\emptyset, \emptyset, \emptyset, \{ \epsilon \})$.

\item For a context gamoid $\Gamma \in \mathcal{PGD}$ and a dependent gamoid $A \in \mathsf{Ty}(\Gamma)$, the comprehension $\Gamma.A \in \mathcal{PGD}$ of $A$ in $\Gamma$ is defined to be the dependent pair space $\widehat{\sum} (\Gamma, A)$.

\item The first projections $\mathsf{p}(A) : \widehat{\sum} (\Gamma, A) \to \Gamma$, where $A$ is a dependent gamoid on $\Gamma$, are the derelictions $\mathsf{der}_\Gamma$ ``up to the tags for disjoint union''.

\item The second projections $\mathsf{v}_A : \widehat{\prod} (\widehat{\sum}(\Gamma, A), A\{\mathsf{p}(A)\})$ are the products 
\begin{equation*}
\textstyle \& \{ \mathsf{v}_\sigma : \gamma \& \sigma \to \sigma \ \! | \ \! \gamma \& \sigma \in \widehat{\sum}(\Gamma, A) \}
\end{equation*}
where $\mathsf{v}_\sigma$ is the dereliction $\mathsf{der}_{\sigma}$ ``up to the tags for disjoint union''.

\item For a dependent gamoid $A$ on a context gamoid $\Gamma \in \mathcal{PGD}$, a morphism $\phi : \Delta \to \Gamma$, and an object $\kappa \in \widehat{\prod}(\Delta, A\{\phi\})$ in $\mathcal{PGD}$, we define the extension $\langle \phi, \kappa \rangle_A : \Delta \to \widehat{\sum} (\Gamma, A)$ as the paring $\phi \& \kappa$ equipped with the arrow-map $p \mapsto \phi_p \& \kappa_p$.
\end{itemize}
\prd
\end{definition}

Of course, we need to establish the following:

\begin{theorem}[Well-defined $\mathcal{PGD}$]
The structure $\mathcal{PGD}$ in fact forms a category with families.
\end{theorem}
\begin{proof}
It is almost straightforward to see that each component is well-defined except for the functions between terms and extensions. For the functions between terms, let $\phi : \Delta \to \Gamma$ be a morphism in $\mathcal{PGD}$, $A : \Gamma \to \mathcal{PGD}$ a dependent gamoid, and $\tau \in \widehat{\prod}(\Gamma,A)$. It has been shown in \cite{yamada2016game} that $\tau \bullet \phi$ is a well-defined strategy on the game $\widehat{\prod}(\Delta, A\{\phi\})$. Moreover, for any $p : \sigma =_\Delta \! \sigma'$, we have $(\tau \bullet \phi)_p = \tau_{\phi_p} : A\phi_p \bullet (\tau \bullet \phi \bullet \sigma) =_{A(\phi \bullet \sigma')} \! \tau \bullet \phi \bullet \sigma'$, i.e., $(\tau \bullet \phi)_p : A\{\phi\}(p) \bullet ((\tau \bullet \phi) \bullet \sigma) =_{A\{\phi\}(\sigma')} \! (\tau \bullet \phi) \bullet \sigma'$. Therefore $\tau \bullet \phi$ is a well-defined object of $\widehat{\prod}(\Delta, A\{\phi\})$, showing that the function $\_\{\phi\}_A : \mathsf{ob}(\widehat{\prod}(\Gamma, A)) \to \mathsf{ob}(\widehat{\prod}(\Delta, A\{\phi\}))$ is well-defined.

Next, we consider extensions. Let $\kappa \in \widehat{\prod}(\Delta, A\{\phi\})$. It has been shown in \cite{yamada2016game} that the object-map of the paring $\phi \& \kappa : \Delta \to \widehat{\sum}(\Gamma, A)$ is well-defined. For the arrow-map, let $r : \sigma_1 =_\Delta \! \sigma_2$. We then have 
\begin{align*}
(\phi \& \kappa)_r  = \phi_r \& \kappa_r : (\phi \bullet \sigma_1) \& (\kappa \bullet \sigma_1) =_{\widehat{\sum}(\Gamma, A)} \! (\phi \bullet \sigma_2) \& (\kappa \bullet \sigma_2)  
\end{align*}
because $\phi_r : \phi \bullet \sigma_1 =_\Gamma \! \phi \bullet \sigma_2$ and $\kappa_r : A\phi_r \bullet (\kappa \bullet \sigma_1) =_{A(\phi \bullet \sigma_2)} \! \kappa \bullet \sigma_2$.
Thus, the arrow-map is well-defined as well.

Finally, we verify the required equations.
\begin{itemize}

\item {\bfseries \sffamily Ty-Id.} For a context gamoid $\Gamma \in \mathcal{PGD}$ and a dependent gamoid $A \in \mathsf{Ty}(\Gamma)$,
\begin{equation*}
A \{ \textsf{id}_\Gamma \} = A \circ \textsf{id}_\Gamma = A.
\end{equation*}

\item {\bfseries \sffamily Ty-Comp.} Additionally, for any composable morphisms $\phi : \Delta \to \Gamma$, $\psi : \Theta \to \Delta$ in $\mathcal{PGD}$,
\begin{equation*}
A \{ \phi \circ \psi \} = A \circ (\phi \circ \psi) = (A \circ \phi) \circ \psi = A \{ \phi \} \{ \psi \}.
\end{equation*}

\item {\bfseries \sffamily Tm-Id.} Moreover, for any object $\tau \in \widehat{\prod}(\Gamma, A)$, 
\begin{equation*}
\tau \{ \textsf{id}_\Gamma \} = \tau \bullet \textsf{id}_\Gamma = \tau.
\end{equation*}

\item {\bfseries \sffamily Tm-Comp.} Under the same assumption,
\begin{equation*}
\tau \{ \phi \circ \psi \} = \tau \bullet (\phi \bullet \psi) = (\tau \bullet \phi) \bullet \psi = \tau \{ \phi \} \{ \psi \}.
\end{equation*}

\item {\bfseries \sffamily Cons-L.} By the definition, we clearly have $\mathsf{p}(A) \bullet \langle \phi, \tau \rangle_A = \phi$.

\item {\bfseries \sffamily Cons-R.} Also, it is immediate that $\mathsf{v}_A \{ \langle \phi, \tau \rangle_A \} = \mathsf{v}_{A} \bullet (\phi \& \tau) = \tau$.

\item {\bfseries \sffamily Cons-Nat.} $\langle \phi, \tau \rangle_A \bullet \psi = (\phi \& \tau) \bullet \psi = (\phi \bullet \psi) \& (\tau \bullet \psi) = \langle \phi \circ \psi, \tau \{ \psi \} \rangle_A$.

\item {\bfseries \sffamily Cons-Id.} Finally, it is clear that $\langle \mathsf{p}(A), \mathsf{v}_A \rangle_A = \mathsf{p}(A) \& \mathsf{v}_A = \mathsf{id}_{\widehat{\sum}(\Gamma, A)} =  \textsf{id}_{\Gamma . A}$.

\end{itemize}

\end{proof}

\subsection{Game-theoretic Type Formers}
We proceed to equip the CwF $\mathcal{PGD}$ with \emph{semantic type formers} (for the definition, see \cite{hofmann1997syntax} or \cite{yamada2016game}), which are categorical structures to interpret specific types such as $\prod$-, $\sum$- and Id-types as well as universes.

\begin{notation}
For brevity, if we have consecutive subscripts such as $(q_p)_r$, then we usually abbreviate it as $q_{p, r}$. We apply the same principle for more than two consecutive subscripts in the obvious way.
\end{notation}

\subsubsection{Game-theoretic $\prod$-types}
We first interpret $\prod$-types, for which we need the following lemma:
\begin{lemma}[Evaluation functor \textsf{ev}]
\label{EV}
Let $B : A \to \mathcal{PGD}$ be a dependent gamoid. Then, we have an object $\mathsf{ev}_{A, B} \in \widehat{\prod}(\widehat{\prod}(A,B) \& A, B\{\mathsf{\mathsf{v}_A}\})$ which is the strategy $\& \{ \zeta \& \sigma \! \rightleftharpoons \! \zeta \bullet \sigma \ \! | \ \! \zeta \& \sigma \in \widehat{\prod}(A,B) \& A \}$ equipped with the arrow-map $(q \& p : \zeta \& \sigma =_{\widehat{\prod}(A,B) \& A} \! \zeta' \& \sigma') \mapsto q_p : Bp \bullet (\zeta \bullet \sigma) =_{B\sigma'} \! \zeta' \bullet \sigma'$.
\end{lemma}

\begin{proof}
First, $\mathsf{ev}_{A,B}$ is clearly a well-defined strategy on the game $\widehat{\prod}(\widehat{\prod}(A,B) \& A, B\{\mathsf{\mathsf{v}_A}\})$, and the arrow-map is well-defined by Proposition \ref{EqualityPreservation}. It remains to verify functoriality of $\mathsf{ev}$. Let $q : \zeta =_{\widehat{\prod}(A,B)} \! \zeta'$, $q' : \zeta' =_{\widehat{\prod}(A,B)} \! \zeta''$, $p : \sigma =_A \! \sigma'$, $p' : \sigma' =_A \! \sigma''$. By naturality of $q, q'$, we have:
\begin{align*}
(\mathsf{ev}_{A, B})_{(q' \! \circ q) \& (p' \! \circ p)} &= (q' \! \circ q)_{\sigma''} \! \circ \tau_{p' \! \circ p} \\
&= (q'_{\sigma''} \! \circ q_{\sigma''}) \circ (\tau_{p'} \! \circ \tau_p) \\
&= q'_{\sigma''} \! \circ (q_{\sigma''} \! \circ \tau_{p'}) \circ \tau_p \\
&= q'_{\sigma''} \! \circ (\tau'_{p'} \! \circ q_{\sigma'}) \circ \tau_p \\
&= (q'_{\sigma''} \! \circ \tau'_{p'}) \circ (q_{\sigma'} \! \circ \tau_p) \\
&= q'_{p'} \! \circ q_p \\
&= (\mathsf{ev}_{A,B})_{q' \& p'} \! \circ (\mathsf{ev}_{A,B})_{q \& p}.
\end{align*}
Also, $(\mathsf{ev}_{A,B})_{\mathsf{id}_{\zeta \& \sigma}} = (\mathsf{ev}_{A,B})_{\mathsf{id}_\zeta \& \mathsf{id}_\sigma} = \mathsf{id}_{\zeta, \mathsf{id}_\sigma} \! = \mathsf{id}_{\zeta, \sigma} \! \circ \zeta_{\mathsf{id}_\sigma} \! = \mathsf{id}_{\zeta \bullet \sigma} \! \circ \mathsf{id}_{\zeta \bullet \sigma} \! = \mathsf{id}_{\zeta \bullet \sigma} \! = \mathsf{id}_{\mathsf{ev}_{A,B} \bullet (\zeta \& \sigma)}$.
\end{proof}

We call $\mathsf{ev}_{A,B}$ the {\bf evaluation functor}.
As mentioned earlier, we required the naturality condition on identifications in dependent function spaces mainly in order to establish this lemma. We shall use the evaluation functors to interpret $\prod$-Elim rule below.

\begin{remark}
The evaluation functors are a generalization of the \emph{evaluation strategies} introduced in \cite{abramsky2000full,mccusker1998games}.
\end{remark}

We are now ready to establish:
\begin{proposition}[$\mathcal{PGD}$ supports $\prod$-types]
The CwF $\mathcal{PGD}$ supports $\prod$-types.
\end{proposition}
\begin{proof}
Let $\Gamma \in \mathcal{PGD}$, $A \in \mathsf{Ty}(\Gamma)$, $B \in \mathsf{Ty}(\widehat{\sum}(\Gamma, A))$ in $\mathcal{PGD}$.
\begin{itemize}

\item {\bfseries \sffamily $\bm{\prod}$-Form.} We need to generalize the construction of dependent function spaces, as $A$ itself is a dependent gamoid. Then we define the dependent gamoid $\prod (A, B) : \Gamma \to \mathcal{PGD}$ by
\begin{align*}
(\gamma : \Gamma) &\mapsto \textstyle \widehat{\prod} (A\gamma, B_\gamma) \\
(p : \gamma =_\Gamma \! \gamma') &\mapsto \textstyle p^{\widehat{\prod}}_{A, B} : \widehat{\prod}(A\gamma, B_\gamma) \to \widehat{\prod}(A\gamma', B_{\gamma'})
\end{align*}
where the dependent gamoid $B_\gamma : A\gamma \to \mathcal{PGD}$ is defined by $B_\gamma(\sigma) \stackrel{\mathrm{df. }}{=} B(\gamma \& \sigma)$ for all $\sigma : A\gamma$, and $B_\gamma(q) \stackrel{\mathrm{df. }}{=} B(\mathsf{id}_\gamma \& q) : B(\gamma \& \sigma_1) \to B(\gamma \& \sigma_2)$ for all $\sigma_1, \sigma_2 : A\gamma$, $q : \sigma_1 =_{A\gamma} \! \sigma_2$, and the morphism $\textstyle p^{\widehat{\prod}}_{A, B} : \widehat{\prod}(A\gamma, B_\gamma) \to \widehat{\prod}(A\gamma', B_{\gamma'})$ in $\mathcal{PGD}$ is defined to be the strategy
\begin{equation*}
\textstyle p^{\widehat{\prod}}_{A, B} \stackrel{\mathrm{df. }}{=} \& \{ \tau \! \rightleftharpoons \! B (p \& \mathsf{id}_{A\gamma'}) \bullet \tau \bullet Ap^{-1} \ \! | \ \! \tau \in \widehat{\prod}(A\gamma, B_\gamma) \}
\end{equation*}
equipped with the arrow-map
\begin{equation*}
(q : \tau_1 =_{\widehat{\prod}(A\gamma, B_\gamma)} \! \tau_2) \mapsto \& \{ B(p \& \mathsf{id}_{A\gamma'})_{q_{Ap^{-1} \bullet \sigma'}} \ \! | \ \! \sigma' : A\gamma' \}.
\end{equation*}
\begin{notation}
We often omit the subscripts $A, B$ in $p^{\widehat{\prod}}_{A, B}$.
\end{notation}
In fact, for any $\tau \in \widehat{\prod}(A\gamma, B_\gamma)$, we have $B(p \& \mathsf{id}_{A\gamma'}) \bullet \tau \bullet Ap^{-1} : A\gamma' \to \uplus B$ that satisfies $(B(p \& \mathsf{id}_{A\gamma'}) \bullet \tau \bullet Ap^{-1}) \bullet \sigma' : B(\gamma' \& \sigma')$ for any $\sigma' : A\gamma'$, and $(p^{\widehat{\prod}} \bullet \tau)_s = (B(p \& \mathsf{id}_{A\gamma'}) \bullet \tau \bullet Ap^{-1})_s : B(p \& \mathsf{id}_{A\gamma'}) \bullet B(\mathsf{id}_\gamma \& (Ap^{-1})_s) \bullet (\tau \bullet Ap^{-1} \bullet \sigma'_1) =_{B(\gamma' \& \sigma'_2)} \! B(p \& \mathsf{id}_{A\gamma'}) \bullet \tau \bullet Ap^{-1} \bullet \sigma'_2$, which implies $(p^{\widehat{\prod}} \bullet \tau)_s : B_{\gamma'}(s) \bullet (p^{\widehat{\prod}} \bullet \tau \bullet \sigma'_1) =_{B_{\gamma'}(\sigma'_2)} \! p^{\widehat{\prod}} \bullet \tau \bullet \sigma'_2$ for all $\sigma'_1, \sigma'_2 : A\gamma'$, $s : \sigma'_1 =_{A\gamma'} \! \sigma'_2$, establishing $p^{\widehat{\prod}} \bullet \tau \in \widehat{\prod}(A\gamma', B_{\gamma'})$. Also for any identification $q : \tau_1 =_{\widehat{\prod}(A\gamma, B_\gamma)} \! \tau_2$, we have
\begin{equation*}
B(p \& \mathsf{id}_{A\gamma'})_{q_{Ap^{-1} \bullet \sigma'}} : B(p \& \mathsf{id}_{A\gamma'}) \bullet \tau_1 \bullet Ap^{-1} \bullet \sigma' =_{B(\gamma' \& \sigma')} B(p \& \mathsf{id}_{A\gamma'}) \bullet \tau_2 \bullet Ap^{-1} \bullet \sigma'
\end{equation*}
i.e., $p^{\widehat{\prod}}_{q, \sigma'} \! : p^{\widehat{\prod}} \bullet \tau_1 \bullet \sigma' =_{B_{\gamma'}(\sigma')} p^{\widehat{\prod}} \bullet \tau_2 \bullet \sigma'$ for all $\sigma' : A\gamma'$; also, naturality of $p^{\widehat{\prod}}_q$ is immediate from that of $q$. Thus, we may conclude that $p^{\widehat{\prod}}_q : p^{\widehat{\prod}} \bullet \tau_1 =_{\widehat{\prod}(A\gamma', B_{\gamma'})} \! p^{\widehat{\prod}}\bullet \tau_2$. Moreover, it is straightforward to see functoriality of $p^{\widehat{\prod}}$. Therefore, we have shown that $p^{\widehat{\prod}}$ is a well-defined morphism $\widehat{\prod}(A\gamma, B_\gamma) \to \widehat{\prod}(A\gamma', B_{\gamma'})$ in $\mathcal{PGD}$.
\begin{remark}
If $\Gamma = \mathcal{D}(I)$, i.e., the empty gamoid, then $\prod (A, B)$ is essentially $\widehat{\prod}(A, B)$; thus $\prod$ is a generalization of $\widehat{\prod}$, and we call $\prod (A, B)$ the {\bf dependent function space} of $B$ over $A$ 
as well.
\end{remark}
 
\item {\bfseries \sffamily $\bm{\prod}$-Intro.} It has been shown in \cite{yamada2016game} that strategies on $\widehat{\prod}(\widehat{\sum} (\Gamma, A), B)$ and strategies on $\widehat{\prod} (\Gamma, \prod (A, B))$ are corresponding ``up to the tags for disjoint union''. Thus, for each strategy $\iota : \widehat{\prod}(\widehat{\sum} (\Gamma, A), B)$, we have the corresponding strategy $\lambda_{A,B}(\iota) : \widehat{\prod} (\Gamma, \prod (A, B))$. It remains to equip $\lambda_{A, B} (\iota)$ with an arrow-map, i.e., we need to establish an identification
\begin{equation*}
\lambda_{A,B}(\iota)_p : B(p \& \mathsf{id}_{A\gamma_2}) \bullet (\lambda_{A,B}(\iota) \bullet \gamma_1) \bullet Ap^{-1} =_{\widehat{\prod}(A\gamma_2, B_{\gamma_2})} \! \lambda_{A,B}(\iota) \bullet \gamma_2
\end{equation*}
for each  $\gamma_1, \gamma_2 : \Gamma$, $p : \gamma_1 =_\Gamma \! \gamma_2$ that makes $\lambda_{A,B}(\iota)$ a functor. By the definition of identifications in dependent function spaces, it suffices to construct, for each $\sigma_2 : A\gamma_2$, an identification
\begin{equation}
\label{LAMBDA}
\lambda_{A,B}(\iota)_{p,\sigma_2} : B(p \& \mathsf{id}_{A\gamma_2}) \bullet (\lambda_{A,B}(\iota) \bullet \gamma_1) \bullet Ap^{-1} \bullet \sigma_2 =_{B(\gamma_2 \& \sigma_2)} \! \lambda_{A,B}(\iota) \bullet \gamma_2 \bullet \sigma_2
\end{equation}
that is natural in $\sigma_2$ and functorial with respect to $p$. Then we have:
\begin{equation}
\label{TAU}
\iota_{p \& \mathsf{id}_{\sigma_2}} : B(p \& \mathsf{id}_{A\gamma_2}) \bullet (\iota \bullet (\gamma_1 \& (Ap^{-1} \bullet \sigma_2))) =_{B(\gamma_2 \& \sigma_2)} \! \iota \bullet (\gamma_2 \& \sigma_2).
\end{equation}
which is clearly functorial with respect to $p$. Also, it is natural in $\sigma_2$ by functoriality of $\iota$, in which note that compositions are made in the dependent union $\uplus B$. Moreover, the games in (\ref{LAMBDA}), (\ref{TAU}) are actually the same. Hence we take $\lambda_{A,B}(\iota)_{p,\sigma_2} \stackrel{\mathrm{df. }}{=} \iota_{p \& \mathsf{id}_{\sigma_2}}$. In addition, it is easy to see that $\lambda_{A,B}(\iota)_{p,q} = \iota_{p \& q}$ for any $\sigma_2, \widehat{\sigma}_2 : A\gamma_2$, $q : \sigma_2 =_{A\gamma_2} \! \widehat{\sigma}_2$ again by functoriality of $\iota$.

\begin{notation}
We often omit the subscripts $A, B$ in $\lambda_{A, B}(\iota)$.
\end{notation}

\item {\bfseries \sffamily $\bm{\prod}$-Elim.} For any $\kappa \in \widehat{\prod} (\Gamma, \prod (A, B))$ and $\tau \in \widehat{\prod} (\Gamma, A)$, we simply define:
\begin{equation*}
\textstyle \mathsf{App}_{A, B} (\kappa, \tau) \stackrel{\mathrm{df. }}{=} \mathsf{ev}_{A, B} \bullet (\kappa \& \tau) 
\end{equation*}
where $\mathsf{ev}_{A,B} \stackrel{\mathrm{df. }}{=} \& \{ \mathsf{ev}_{A\gamma, B_\gamma} \ \! | \ \! \gamma : \Gamma \} : \& \{ \widehat{\prod}(A\gamma, B_{\gamma}) \& A\gamma \to B_{\gamma}\{ \mathsf{v}_{A}\} \ \! | \ \! \gamma : \Gamma \}$. It is then easy to see that $\mathsf{App}_{A, B} (\kappa, \tau)$ is a strategy on the game $\widehat{\prod}(\Gamma, B\{\overline{\tau}\})$, where $\overline{\tau} \stackrel{\mathrm{df. }}{=} \mathsf{id}_\Gamma \& \tau$.
Also for the arrow-map, for any $\gamma, \gamma' : \Gamma$, $p : \gamma =_{\Gamma} \! \gamma'$, we have an identification
\begin{equation*}
\textstyle \kappa_p : \prod(A,B)(p) \bullet (\kappa \bullet \gamma) =_{\widehat{\prod}(A\gamma', B_{\gamma'})} \! \kappa \bullet \gamma'
\end{equation*}
which in turn induces an identification
\begin{equation*}
\kappa_{p, \tau_p} : B(\textsf{id}_{\gamma'} \& \tau_p)\bullet (p^{\widehat{\prod}} \bullet (\kappa \bullet \gamma) \bullet Ap(\tau \bullet \gamma)) =_{B( \gamma' \& (\tau \bullet \gamma'))} \kappa \bullet \gamma' \bullet \tau \bullet \gamma'
\end{equation*}
where $\kappa_{p, \tau_p} = \mathsf{App}_{A, B} (\kappa, \tau)_p$, and
\begin{align*}
& B(\textsf{id}_{\gamma'} \& \tau_p) \bullet (p^{\widehat{\prod}} \bullet (\kappa \bullet\gamma) \bullet Ap \bullet (\tau \bullet \gamma)) \\
= \ & B(\textsf{id}_{\gamma'} \& \tau_p) \bullet B(p \& \textsf{id}_{Ap \bullet (\tau \bullet \gamma)}) \bullet \kappa \bullet \gamma \bullet Ap^{-1} \bullet Ap \bullet (\tau \bullet \gamma) \\
= \ & B(p \& \tau_p) \bullet (\kappa \bullet \gamma \bullet \tau \bullet \gamma)
\end{align*}
establishing an identification
\begin{equation*}
\mathsf{App}_{A, B} (\kappa, \tau)_p : B(p \& \tau_p) \bullet (\mathsf{App}_{A, B} (\kappa, \tau) \bullet \gamma) =_{B(\gamma' \& (\tau \bullet \gamma'))} \! \mathsf{App}_{A, B} (\kappa, \tau) \bullet \gamma'.
\end{equation*}
Moreover, $\mathsf{App}_{A,B}(\kappa, \tau)$ is functorial by Lemma \ref{EV} (i.e., because it is a composition of functors).
Hence we may conclude that $\mathsf{App}_{A, B} (\kappa, \tau) \in \widehat{\prod} (\Gamma, B\{ \overline{\tau} \})$.

\begin{notation}
We often omit the subscripts $A, B$ in $\mathsf{App}_{A, B}(\kappa, \tau)$.
\end{notation}

\item {\bfseries \sffamily $\bm{\prod}$-Comp.} Let $\iota \in \widehat{\prod}(\widehat{\sum}(\Gamma, A), B)$, $\tau \in \widehat{\prod}(\Gamma, A)$. First, it is straightforward to see that $\mathsf{App} (\lambda (\iota) , \tau), \iota \{ \overline{\tau} \}$ are the same strategies.
And for the arrow-maps, we have
\begin{align*}
\textstyle \mathsf{App} (\lambda (\iota) , \tau)_p &= \lambda (\iota)_{p, \tau_p} \\
&= \iota_{p \& \tau_p} \\
&= \iota_{(\mathsf{id}_\Gamma \& \tau)_p} \\
&= \iota \{ \overline{\tau} \}_p.
\end{align*}

Hence we may conclude that $\mathsf{App} (\lambda (\sigma) , \tau) = \sigma \{ \overline{\tau} \}$.

\item {\bfseries \sffamily $\bm{\prod}$-Subst.} Moreover, for any $\Delta \in \mathcal{PGD}$ and $\phi : \Delta \to \Gamma$ in $\mathcal{PGD}$, we have, for the object-map,
\begin{align*}
\textstyle \prod (A, B) \{ \phi \}(\delta) &= \textstyle  \widehat{\prod} (A(\phi \bullet \delta), B_{\phi \bullet \delta})  \\
&= \textstyle \widehat{\prod} (A\{\phi\}(\delta), B\{\phi^+\}_\delta) \\
&= \textstyle \prod (A\{\phi\}, B\{\phi^+\})(\delta)
\end{align*}
for all $\delta : \Delta$, where $\phi^+  \stackrel{\mathrm{df. }}{=} (\phi \bullet \mathsf{p}(A \{ \phi \})) \& \mathsf{v}_{A\{\phi\}} : \widehat{\sum} (\Delta,  A\{\phi\}) \to \widehat{\sum} (\Gamma, A)$. Note that the second equation holds because
\begin{align*}
B\{\phi^+\}_\delta (\psi) &= B\{\phi^+\} (\delta \& \psi) \\ 
&= B (((\phi \bullet \mathsf{p}(A \{ \phi \})) \& \mathsf{v}_{A\{\phi\}}) \bullet (\delta \& \psi)) \\
&= B ( (\phi \bullet \mathsf{p}(A \{ \phi \}) \bullet (\delta \& \psi)) \& (\mathsf{v}_{A\{\phi\}} \bullet (\delta \& \psi))) \\
&= B ( (\phi \bullet \delta) \& \psi ) \\
&= B_{\phi \bullet \delta} (\psi)
\end{align*}
for all $\psi : A(\phi \bullet \gamma)$, and similarly $B\{\phi^+\}_\delta (q) = B_{\phi \bullet \delta} (q)$ for all identifications $q$ in $A(\phi \bullet \gamma)$.
And for the arrow-map, it is not hard to see, for any $p : \delta=_\Delta \! \delta'$, that
\begin{align*}
\textstyle \prod (A, B) \{ \phi \}(p) &= \textstyle \prod (A, B)(\phi_p) \\
&= \textstyle \& \{ \tau \! \rightleftharpoons \! B(\phi_p \& \mathsf{id}_{A (\phi \bullet \delta')}) \bullet \tau \bullet A\phi_p^{-1} \ \! | \ \! \tau \in \widehat{\prod}(A(\phi \bullet \delta), B_{\phi \bullet \delta}) \} \\
&= \textstyle \& \{ \tau \! \rightleftharpoons \! B\{\phi^+\}(p \& \mathsf{id}_{A\{\phi\}(\delta')}) \bullet \tau \bullet A\{\phi\}_{p^{-1}} \ \! | \ \! \tau \in \widehat{\prod}(A\{\phi\}(\delta), B\{\phi^+\}_{\delta}) \} \\
&= \textstyle \prod(A\{\phi\}, B\{ \phi^+ \})(p).
\end{align*}

\item {\bfseries \sffamily $\bm{\lambda}$-Subst.} For any object $\iota \in \widehat{\prod} (\widehat{\sum} (\Gamma, A), B)$, it is easy to see that $\lambda (\iota) \{ \phi \}, \lambda (\iota \{ \phi^+ \})$ are the same strategies. For the arrow-maps, for any $\delta, \delta' : \Delta$, $p : \delta =_\Delta \! \delta'$, $\tau, \tau' : \prod(A, B)\{\phi\}$, $q : \prod(A, B)\{\phi\}(p) \bullet \tau \bullet \delta =_{\widehat{\prod}(A(\phi \bullet \delta'), B_{\phi \bullet \delta'})} \! \tau' \bullet \delta'$, observe that:
\begin{align*}
\lambda (\iota) \{ \phi \}_{p, q} = \lambda(\iota)_{\phi_p, q} = \iota_{\phi_p \& q} = \iota\{\phi^+\}_{p \& q} = \lambda (\iota\{\phi^+\})_{p, q}
\end{align*}
establishing $\lambda (\iota) \{ \phi \} = \lambda (\iota \{ \phi^+ \})$.

\item {\bfseries \sffamily App-Subst.} Finally, we have:
\begin{align*}
\mathsf{App} (\kappa, \tau) \{ \phi \} &= (\mathsf{ev} \bullet (\kappa \& \tau)) \bullet \phi  \\
&= \mathsf{ev} \bullet ((\kappa \& \tau) \bullet \phi)  \\
&= \mathsf{ev} \bullet ((\kappa \bullet \phi) \& (\tau \bullet \phi)) \\
&= \mathsf{ev} \bullet (\kappa \{ \phi \} \& \tau \{ \phi \})  \\
&= \mathsf{App} (\kappa \{ \phi \}, \tau \{ \phi \})
\end{align*}
where we omit the subscripts $A, B$ in $\mathsf{ev}_{A, B}$.

\end{itemize}

\end{proof}

\subsubsection{Game-theoretic $\sum$-types}
Next, we handle $\sum$-types.
\begin{proposition}[$\mathcal{PGD}$ supports $\sum$-types]
The CwF $\mathcal{PGD}$ supports $\sum$-types.
\end{proposition}
\begin{proof}
Let $\Gamma \in \mathcal{PGD}$, $A \in \mathsf{Ty}(\Gamma)$, and $B \in \mathsf{Ty}(\widehat{\sum}(\Gamma, A))$ in $\mathcal{PGD}$.
\begin{itemize}
\item {\bfseries \sffamily $\bm{\sum}$-Form.} Similar to the case of dependent function spaces, we generalize the dependent pair space construction $\widehat{\sum}$ to the construction $\sum$ as follows. We define the dependent gamoid $\sum(A, B) : \Gamma \to \mathcal{PGD}$ by:
\begin{align*}
\textstyle (\gamma : \Gamma) &\mapsto \textstyle \widehat{\sum}(A\gamma, B_\gamma) \\
(p : \gamma =_\Gamma \! \gamma') &\mapsto \textstyle p^{\widehat{\sum}} : \widehat{\sum}(A\gamma, B_\gamma) \to \widehat{\sum}(A\gamma', B_{\gamma'})
\end{align*}
where the functor $p^{\widehat{\sum}} : \widehat{\sum}(A\gamma, B_\gamma) \to \widehat{\sum}(A\gamma', B_{\gamma'})$ is defined to be the strategy
\begin{equation*}
p^{\widehat{\sum}} \stackrel{\mathrm{df. }}{=}  \& \{ \sigma \& \tau \! \rightleftharpoons \! (Ap \bullet \sigma) \& (B( p \& \textsf{id}_{Ap \bullet \sigma}) \bullet \tau) \ \! | \ \! \sigma \& \tau \in \textstyle \widehat{\sum}(A\gamma, B_\gamma) \}
\end{equation*}
equipped with the arrow-map
\begin{align*}
s \& t &\mapsto (Ap)_s \& B (p \& \textsf{id}_{Ap \bullet \sigma_2})_t
\end{align*}
for any $\sigma_1, \sigma_2 : A\gamma$, $\tau_1 : B (\gamma \& \sigma_1)$, $\tau_2 : B (\gamma \& \sigma_2)$, $s : \sigma_1 =_{A\gamma} \! \sigma_2$, $t : B ( \textsf{id}_\gamma \& s) \bullet \tau_1 =_{B (\gamma \& \sigma_2)} \! \tau_2$. Note that we have:
\begin{align*}
s \& t &: \sigma_1 \& \tau_1 =_{\widehat{\sum}(A\gamma, B_\gamma)} \!\sigma_2 \& \tau_2 \\
(Ap)_s &: Ap \bullet \sigma_1 =_{A\gamma'} \! Ap \bullet \sigma_2 \\
B ( p \& \textsf{id}_{Ap \bullet \sigma_2} )_t &: B(\mathsf{id}_{\gamma'} \& (Ap)_s) \bullet B(p \& \mathsf{id}_{Ap \bullet \sigma_1}) \bullet \tau_1 =_{B(\gamma' \& Ap \bullet \sigma_2)} \! B (p \& \textsf{id}_{Ap \bullet \sigma_2})  \bullet \tau_2 \\
(Ap)_s \& B (p \& \textsf{id}_{Ap \bullet \sigma_2})_t &: (Ap \bullet \sigma_1) \& (B(p \& \mathsf{id}_{Ap \bullet \sigma_1}) \bullet \tau_1) =_{\widehat{\sum}(A\gamma', B_{\gamma'})} \! (Ap \bullet \sigma_2) \& (B(p \& \mathsf{id}_{Ap \bullet \sigma_2}) \bullet \tau_2). 
\end{align*}
Thus, the arrow-map is well-defined. Also, it is easy to observe its functoriality.

\item {\bfseries \sffamily $\bm{\sum}$-Intro.} We need the following lemma:
\begin{lemma}[$\widehat{\sum} \sum$-correspondence lemma] 
\label{SumSum}
For any dependent gamoids $A : \Gamma \to \mathcal{PGD}$, $B : \widehat{\sum}(\Gamma, A) \to \mathcal{PGD}$, we have a correspondence
\begin{equation*}
\textstyle \widehat{\sum}(\widehat{\sum}(\Gamma,A),B) \cong \widehat{\sum}(\Gamma, \sum(A,B))
\end{equation*}
\end{lemma}
\begin{proof}[Proof of the lemma]
Because the correspondence between objects is obvious, it suffices to establish the correspondence between identifications. First, note that an identification between objects $(\gamma \& \sigma) \& \tau, (\gamma' \& \sigma')  \& \tau' \in \textstyle \widehat{\sum}(\widehat{\sum}(\Gamma,A),B)$ is of the form
\begin{equation*}
(p \& q) \& s : (\gamma \& \sigma) \& \tau =_{\textstyle \widehat{\sum}(\widehat{\sum}(\Gamma,A),B)} \! (\gamma' \& \sigma') \& \tau'
\end{equation*}
where $p : \gamma =_{\Gamma} \! \gamma'$, $q : Ap \bullet \sigma =_{A\gamma'} \! \sigma'$ and $s : B(p \& q) \bullet \tau =_{B (\gamma' \& \sigma')} \! \tau'$. 
It is then not hard to see that there is the corresponding identification
\begin{equation*}
p \& (q \& s) : \gamma \& (\sigma \& \tau) =_{\textstyle \widehat{\sum}(\Gamma, \sum(A,B))} \! \gamma' \& (\sigma' \& \tau').
\end{equation*}
Moreover, we can clearly reverse this process, completing the proof of the lemma.
\end{proof}

Thanks to the lemma, we then define the pair
\begin{equation*}
\textstyle \mathsf{Pair}_{A, B} : \widehat{\sum}(\widehat{\sum}(\Gamma,A),B) \to \widehat{\sum}(\Gamma, \sum(A,B))
\end{equation*}
as the obvious dereliction ``up to the tag of moves for disjoint union''.

\item {\bfseries \sffamily $\bm{\sum}$-Elim.} For any $P \in \mathsf{Ty}(\widehat{\sum}(\Gamma, \sum(A,B)))$ and $\psi \in \widehat{\prod}(\widehat{\sum}(\widehat{\sum}(\Gamma, A), B), P\{\mathsf{Pair}_{A,B}\})$, by Lemma~\ref{SumSum}, we may construct an object
\begin{equation*}
\textstyle R^{\sum}_{A, B, P}(\psi) \in \widehat{\prod}(\widehat{\sum}(\Gamma, \sum (A, B)), P) 
\end{equation*}
from $\psi$ just by ``adjusting the tags of moves for disjoint union''.

\item {\bfseries \sffamily $\bm{\sum}$-Comp.} The equation
\begin{align*}
\textstyle R^{\sum}_{A, B, P}(\psi) \{ \mathsf{Pair}_{A, B}\} = \textstyle R^{\sum}_{A, B, P}(\psi) \bullet \mathsf{Pair}_{A, B} = \psi
\end{align*}
is obvious by the definition.

\item {\bfseries \sffamily $\bm{\sum}$-Subst.} Moreover, for any context gamoid $\Delta \in \mathcal{PGD}$ and a morphism $\phi : \Delta \to \Gamma$ in $\mathcal{PGD}$, by the same calculation as the case of dependent function spaces, we have, for the object-map,
\begin{align*}
\textstyle \sum (A, B) \{ \phi \}(\delta) &= \textstyle \widehat{\sum}(A(\phi \bullet \delta), B_{\phi \bullet \delta}) \\
&= \textstyle \widehat{\sum}(A\{\phi\}(\delta), B\{\phi^+\}_\delta) \\
&= \textstyle \sum (A\{\phi\}, B\{\phi^+\})(\delta)
\end{align*}
for all $\delta : \Delta$, where $\phi^+ \stackrel{\mathrm{df. }}{=} (\phi \bullet \mathsf{p}(A\{ \phi \}) \& \mathsf{v}_{A\{\phi\}}) : \widehat{\sum}(\Delta, A\{\phi\}) \to \widehat{\sum}(\Gamma, A)$. 
And for the arrow-map, for any $\delta, \delta' : \Delta$, $p : \delta =_\Delta \! \delta$, we clearly have
\begin{align*}
& \ \textstyle \sum (A, B) \{ \phi \}(p) \\
= & \ \textstyle \sum (A, B)(\phi_p) \\
= & \ \& \{ \textstyle \sigma \& \tau \! \rightleftharpoons \! (A\phi_p \bullet \sigma) \& (B(\phi_p \& \mathsf{id}_{A\phi_p \bullet \sigma}) \bullet \tau) \ \! | \ \! \sigma \& \tau \in \widehat{\sum}(A(\phi \bullet \delta), B_{\phi \bullet \delta}) \} \\
= & \ \textstyle \& \{ \sigma \& \tau \! \rightleftharpoons \! (A\{\phi\}(p) \bullet \sigma) \& (B\{\phi^+\}(p \& \mathsf{id}_{A\{\phi\}(p) \bullet \sigma}) \bullet \tau) \ \! | \ \! \sigma \& \tau \in \widehat{\sum}(A\{\phi\}(\delta), B\{\phi^+\}_\delta) \} \\
= & \ \textstyle \sum (A\{\phi\}, B\{\phi^+\})(p).
\end{align*}
Therefore, we may conclude that $\sum (A, B) \{ \phi \} = \sum (A\{\phi\}, B\{\phi^+\})$.

\item {\bfseries \sffamily Pair-Subst.} Under the same assumption, the equation
\begin{equation*}
\textstyle \mathsf{p}(\sum (A, B)) \bullet \mathsf{Pair}_{A, B} = \mathsf{p}(A) \bullet \mathsf{p}(B)
\end{equation*}
is obvious by the definition, and we also have:
\begin{align*}
& \ \phi^* \bullet \mathsf{Pair}_{A\{\phi\}, B\{\phi^+\}} \\
= & \ (\phi \bullet \textstyle \mathsf{p}(\sum(A,B)\{\phi\})) \& \mathsf{v}_{\sum(A,B)\{\phi\}} ) \bullet \mathsf{Pair}_{A\{\phi\}, B\{\phi^+\}} \\
= & \ (\phi \bullet \textstyle \mathsf{p}(\sum(A\{\phi\},B\{\phi^+\})) \bullet \mathsf{Pair}_{A\{\phi\} \& (B\{\phi^+\}}) \& (\mathsf{v}_{\sum(A,B)\{\phi\}} \bullet \mathsf{Pair}_{A\{\phi\}, B\{\phi^+\}}) \\
= & \ (\phi \bullet \mathsf{p}(A\{\phi\}) \bullet \mathsf{p}(B\{\phi^+\})) \& (\mathsf{v}_{\sum(A\{\phi\},B\{\phi^+\})} \bullet \mathsf{Pair}_{A\{\phi\}, B\{\phi^+\}}) \\
= & \ \mathsf{Pair}_{A, B} \bullet ((\phi^+ \bullet \mathsf{p}(B\{\phi^+\})) \& \mathsf{v}_{B\{\phi^+\}}) \\
= & \ \mathsf{Pair}_{A, B} \bullet \phi^{++} 
\end{align*}
where $\phi^* \stackrel{\mathrm{df. }}{=} (\phi \bullet \mathsf{p}(\sum(A,B)\{\phi\})) \& \mathsf{v}_{\sum(A,B)\{\phi\}} : \widehat{\sum}(\Delta, \sum(A\{\phi\}, B\{\phi^+\})) \to \widehat{\sum}(\Gamma, \sum(A, B))$, $\phi^{++} \stackrel{\mathrm{df. }}{=} (\phi^+ \bullet \mathsf{p}(B\{\phi^+\})) \& \mathsf{v}_{B\{\phi^+\}} : \widehat{\sum}(\widehat{\sum}(\Gamma, A), B\{\phi^+\}) \to \widehat{\sum}(\widehat{\sum}(\Gamma, A), B)$.

\item {\bfseries \sffamily $\bm{R^{\sum}}$-Subst.} Finally, we have:
\begin{align*}
& \textstyle R^{\sum}_{A, B, P}(\psi) \{\phi^*\} \\
= \ & \textstyle R^{\sum}_{A, B, P}(\psi) \bullet ((\phi \bullet \mathsf{p}(\sum(A,B)\{\phi\})) \& \mathsf{v}_{\sum(A,B)\{\phi\}}) \\
= \ & R^{\sum}_{A\{\phi\}, B\{\phi^+\}, P\{\phi^*\}} (\psi \bullet ((\phi^+ \bullet \mathsf{p}(B\{\phi^+\})) \& \mathsf{v}_{B\{\phi^+\}})) \\
= \ & R^{\sum}_{A\{\phi\}, B\{\phi^+\}, P\{\phi^*\}} (\psi \bullet \phi^{++} ) \\
= \ & R^{\sum}_{A\{\phi\}, B\{\phi^+\}, P\{\phi^*\}} (\psi \{ \phi^{++} \}) 
\end{align*}
as a strict equality between strategies. And for the arrow-map, we have:
\begin{align*}
\textstyle R^{\sum}_{A, B, P}(\psi) \{\phi^*\}_{p \& (q \& s)} &= R^{\sum}_{A, B, P}(\psi)_{\phi_p \& (q \& s)} \\
&= R^{\sum}_{A, B, P}(\psi_{(\phi_p \& q) \& s}) \\
&= R^{\sum}_{A\{\phi\}, B\{\phi^+\}, P\{\phi^*\}} (\psi \{ \phi^{++} \}_{(p \& q) \& s}) \\
&= R^{\sum}_{A\{\phi\}, B\{\phi^+\}, P\{\phi^*\}} (\psi \{ \phi^{++} \})_{p \& (q \& s)} 
\end{align*}
for any equality $p \& (q \& s) : \delta \& (\sigma \& \tau) =_{\widehat{\sum} (\Delta, \sum(A, B)\{ \phi \})} \! \delta' \& (\sigma' \& \tau')$ in $\widehat{\sum} (\Delta, \sum(A, B)\{ \phi \})$.
Thus, we may conclude that $R^{\sum}_{A, B, P}(\psi) \{\phi^*\} = R^{\sum}_{A\{\phi\}, B\{\phi^+\}, P\{\phi^*\}} (\psi \{ \phi^{++} \})$.
\end{itemize}

\end{proof}

\subsubsection{Game-theoretic Id-types}
Now, we consider \emph{intensional} Id-types. Our interpretation will refute the principle of \emph{uniqueness of identity proofs} and validate the \emph{univalence axiom} as well as the axiom of \emph{function extensionality} (see Section~\ref{Intensionality}), though we do not interpret non-trivial higher propositional equalities.

\begin{proposition}[$\mathcal{PGD}$ supports intensional Id-types]
The CwF $\mathcal{PGD}$ of gamoids supports the (intensional) identity types.
\end{proposition}
\begin{proof}
Let $\Gamma \in \mathcal{PGD}$, $A \in \mathsf{Ty}(\Gamma)$, and $A^+ \stackrel{\mathrm{df. }}{=} A\{\mathsf{p}(A)\} \in \mathsf{Ty}(\widehat{\sum}(\Gamma, A))$.
\begin{itemize}

\item {\bfseries \sffamily Id-Form.} The dependent gamoid $\textsf{Id}_A : \widehat{\sum}(\widehat{\sum}(\Gamma, A), A^+) \to \mathcal{PGD}$ is defined by
\begin{align*}
\textstyle (\gamma \& \sigma_1) \& \sigma_2 \in \widehat{\sum}(\widehat{\sum}(\Gamma, A), A^+)) &\mapsto \widehat{\textsf{Id}}_A(\sigma_1, \sigma_2)
\end{align*}
for the object-map, and 
\begin{align*}
(p \& q_1) \& q_2 \mapsto \widehat{\textsf{Id}}_A(q_1,q_2) :  \widehat{\textsf{Id}}_A(\sigma_1, \sigma_2) \to \widehat{\textsf{Id}}_A(\sigma'_1, \sigma'_2)
\end{align*}
for the arrow-map, where $p : \gamma =_{\Gamma} \! \gamma'$, $\sigma_1, \sigma_2 : A\gamma$, $\sigma'_1, \sigma'_2 : A\gamma'$, $q_1 : Ap \bullet \sigma_1 =_{A\gamma'} \! \sigma'_1$, $q_2 : Ap \bullet \sigma_2 =_{A\gamma'} \! \sigma'_2$, and the morphism $\widehat{\textsf{Id}}_A(q_1,q_2) :  \widehat{\textsf{Id}}_A(\sigma_1, \sigma_2) \to \widehat{\textsf{Id}}_A(\sigma'_1, \sigma'_2)$ in $\mathcal{PGD}$ is defined to be the strategy
\begin{equation*}
\& \{ \alpha \! \rightleftharpoons \! q_2 \bullet (Ap)_\alpha \bullet q_1^{-1} \ \! | \ \! \alpha \in \widehat{\textsf{Id}}_A(\sigma_1, \sigma_2) \}
\end{equation*}
equipped with the arrow-map $\textsf{id}_\alpha \mapsto \textsf{id}_{q_2 \bullet (Ap)_\alpha \bullet q_1^{-1}}$ for all $\alpha \in \widehat{\textsf{Id}}_A(\sigma_1, \sigma_2)$.

\item {\bfseries \sffamily Id-Intro.} The morphism
\begin{equation*}
\textstyle \textsf{Refl}_A : \widehat{\sum}(\Gamma, A_1) \to \widehat{\sum}(\widehat{\sum}(\widehat{\sum}(\Gamma, A_2), A_3^+), \textsf{Id}_A)
\end{equation*}
is defined to be the copy-cat strategy between $\widehat{\sum}(\Gamma, A_1)$ and $\widehat{\sum}(\Gamma, A_2)$, $A_1$ and $A_3^+$, or on $\textsf{Id}_A$, where the subscripts are to distinguish the different copies of $A$, i.e., 
\begin{equation*}
\textstyle \textsf{Refl}_A \stackrel{\mathrm{df. }}{=} \& \{ \gamma \& \sigma \! \rightleftharpoons ((\gamma \& \sigma) \& \sigma) \& \mathsf{id}_\sigma \ \! | \ \! \gamma \& \sigma \in \widehat{\sum} (\Gamma, A) \}
\end{equation*}
equipped with the arrow-map $(p \& q : \gamma \& \sigma =_{\widehat{\sum}(\Gamma, A)} \! \gamma' \& \sigma') \mapsto ((p \& q) \& q) \& \mathsf{id}_{\mathsf{id}_{\sigma'}}$.

\item {\bfseries \sffamily Id-Elim.} For a dependent gamoid $B \in \mathsf{Ty}(\widehat{\sum}(\widehat{\sum}(\widehat{\sum}(\Gamma, A),  A^+), \textsf{Id}_A))$ and an object $\tau \in \widehat{\prod}(\widehat{\sum}(\Gamma, A), B\{\textsf{Refl}_A\})$ in $\mathcal{PGD}$, we need to define an object
\begin{equation*}
\textstyle R^{\textsf{Id}}_{A,B}(\tau) \in \widehat{\prod}(\widehat{\sum}(\widehat{\sum}(\widehat{\sum}(\Gamma, A_1), A_2^+), \textsf{Id}_A), B).
\end{equation*}
Then, note that we have an equality 
\begin{align*}
(\textsf{id}_{\gamma \& \sigma_1} \& \alpha) \& \textsf{id}_\alpha : \textsf{Refl}_A \bullet (\gamma \& \sigma_1)  =_{\widehat{\sum}(\widehat{\sum}(\widehat{\sum}(\Gamma, A_1), A_2^+), \textsf{Id}_A)} \! ((\gamma \& \sigma_1) \& \sigma_2) \& \alpha
\end{align*}
because $\textsf{id}_{\gamma \& \sigma_1} : \gamma \& \sigma_1 =_{\widehat{\sum}(\Gamma, A_1)} \! \gamma \& \sigma_1 $, $\alpha : A^+(\mathsf{id}_{\gamma \& \sigma_1}) \bullet \sigma_1 =_{A^+(\gamma \& \sigma_2)} \! \sigma_2$, and $\textsf{id}_\alpha : \textsf{Id}_{A}(\textsf{id}_{\gamma \& \sigma_1} \& \alpha) \bullet \textsf{id}_{\sigma_1} =_{\mathsf{Id}_A((\gamma \& \sigma_1) \& \sigma_2)} \! \alpha$. This induces the isomorphism functor
\begin{equation*}
B^{\textsf{Id}}_\alpha \stackrel{\mathrm{df. }}{=} B((\textsf{id}_{ \gamma \& \sigma_1} \& \alpha) \& \textsf{id}_\alpha) : B(\textsf{Refl}_A \bullet (\gamma \& \sigma_1)) \stackrel{\simeq}{\to} B(((\gamma \& \sigma_1) \& \sigma_2)  \& \alpha).
\end{equation*}
We then define $R^{\textsf{Id}}_{A,B}(\tau)$ to be
\begin{equation*}
\textstyle  \& \{ ((\gamma \& \sigma_1) \& \sigma_2) \& \alpha \! \rightleftharpoons \! B^{\textsf{Id}}_\alpha \bullet \tau \bullet (\gamma \& \sigma_1) \ \! | \ \! ((\gamma \& \sigma_1) \& \sigma_2) \& \alpha \in \widehat{\sum}(\widehat{\sum}(\widehat{\sum}(\Gamma, A_1), A_2^+), \textsf{Id}_A) \}.
\end{equation*} 
\begin{notation}
We often omit the subscripts $A, B$ in $R^{\textsf{Id}}_{A,B}(\tau)$.
\end{notation}

\item {\bfseries \sffamily Id-Comp.} By the definition, it is straightforward to see that $R^{\textsf{Id}}_{A,B}(\tau)\{\textsf{Refl}_A\} = \tau$:
\begin{align*}
R^{\textsf{Id}}_{A,B}(\tau)\{\textsf{Refl}_A\} &= R^{\textsf{Id}}_{A,B}(\tau) \bullet \textsf{Refl}_A \\
&= \textstyle \& \{ \gamma \& \sigma \! \rightleftharpoons \! B^{\textsf{Id}}_{\textsf{id}_\sigma} \! \bullet \tau \bullet (\gamma \& \sigma) \ \! | \ \! \gamma \& \sigma \in \widehat{\sum}(\Gamma, A) \} \\
&= \textstyle \& \{ \gamma \& \sigma \! \rightleftharpoons \! \tau \bullet (\gamma \& \sigma) \ \! | \ \! \gamma \& \sigma \in \widehat{\sum}(\Gamma, A) \} \\
&= \tau.
\end{align*}

\item {\bfseries \sffamily Id-Subst.} Furthermore, for any context gamoid $\Delta \in \mathcal{PGD}$ and strategy $\phi : \Delta \to \Gamma$ in $\mathcal{PGD}$, it is straightforward to see $\textsf{Id}_A\{\phi^{++}\} =  \textsf{Id}_{A\{\phi\}}$: For the object-map, we have
\begin{align*}
\textsf{Id}_A\{\phi^{++}\}((\delta \& \sigma_1) \& \sigma_2) &= \textstyle \textsf{Id}_{A\{\phi\}}(\phi^{++} \bullet (\delta \& \sigma_1) \& \sigma_2) \\
&= \textstyle \textsf{Id}_{A\{\phi\}}(((\phi \bullet \delta) \& \sigma_1) \& \sigma_2) \\
&= \textstyle \widehat{\textsf{Id}}_{A\{\phi\}}(\sigma_1, \sigma_2) \\
&= \textsf{Id}_{A\{\phi\}} ((\delta \& \sigma_1) \& \sigma_2)
\end{align*}
for any $(\delta \& \sigma_1) \& \sigma_2 \in \widehat{\sum}(\widehat{\sum}(\Delta, A\{\phi\}), A\{\phi\}^+)$, where $\phi^+ \stackrel{\mathrm{df. }}{=} (\phi \bullet \mathsf{p}(A\{ \phi \})) \& \mathsf{v}_{A\{\phi\}} : \widehat{\sum}(\Delta, A\{\phi\}) \to \widehat{\sum}(\Gamma, A)$, and $\phi^{++} \stackrel{\mathrm{df. }}{=} (\phi^+ \bullet \mathsf{p}(A^+\{\phi^+\})) \& \mathsf{v}_{A^+\{\phi^+\}} : \widehat{\sum}(\widehat{\sum}(\Delta, A\{\phi\}), A^+\{\phi^+\}) \to \widehat{\sum}(\widehat{\sum}(\Gamma, A), A^+)$, and for the arrow-map,
\begin{align*}
\textsf{Id}_A\{\phi^{++}\}((p \& q) \& s) &= \textstyle \textsf{Id}_{A\{\phi\}}(\phi^{++}_{(p \& q) \& s}) \\
&= \textstyle \textsf{Id}_{A\{\phi\}}((\phi_p \& q) \& s) \\
&= \textstyle \widehat{\textsf{Id}}_{A\{\phi\}}(q, s) \\
&= \textsf{Id}_{A\{\phi\}} ((p \& q) \& s)
\end{align*}
for any identification $(p \& q) \& s$ in $\widehat{\sum}(\widehat{\sum}(\Delta, A\{\phi\}), A\{\phi\}^+)$.

\item {\bfseries \sffamily Refl-Subst.} Also, the following equation holds:
\begin{align*}
\textsf{Refl}_A \bullet \phi^+ &= \textsf{Refl}_A \bullet ((\phi \bullet \mathsf{p}(A\{ \phi \})) \& \mathsf{v}_{A\{\phi\}}) \\
&= \textstyle \& \{ \delta \& \sigma \! \rightleftharpoons \! (((\phi \bullet \delta) \& \sigma) \& \sigma) \& \mathsf{id}_{\sigma} \ \! | \ \! \delta \& \sigma \in \widehat{\sum} (\Delta, A\{\phi\}) \} \\
&= ((\phi^{++} \bullet \mathsf{p}(\textsf{Id}_A\{\phi^{++}\})) \& \mathsf{v}_{\textsf{Id}_A\{\phi^{++}\}}) \bullet \textsf{Refl}_{A\{\phi\}}  \\
&= \phi^{+++} \bullet \textsf{Refl}_{A\{\phi\}} 
\end{align*}
where $\phi^{+++}  \stackrel{\mathrm{df. }}{=} (\phi^{++} \! \bullet p(\textsf{Id}_A\{\phi^{++}\})) \& v_{\textsf{Id}_A\{\phi^{++}\}}$.

\item {\bfseries \sffamily $\bm{R^{\textsf{Id}}}$-Subst.} Finally, we have:
\begin{align*}
& \ R^{\textsf{Id}}_{A,B}(\tau)\{\phi^{+++}\} \\
= & \ (R^{\textsf{Id}}_{A,B}(\tau) \bullet (\phi^{++} \bullet p(\textsf{Id}_A\{\phi^{++}\})) \& v_{\textsf{Id}_A\{\phi^{++}\}}) \\
= & \ \textstyle \& \{ ((\delta \& \sigma_1) \& \sigma_2) \& \alpha \! \rightleftharpoons \! B^{\mathsf{Id}}_{\alpha} \bullet \tau \bullet ((\phi \bullet \delta) \& \sigma_1) \ \! | \ \! ((\delta \& \sigma_1) \& \sigma_2) \& \alpha \in \widehat{\sum} (\widehat{\sum} (\widehat{\sum} (\Delta, A\{\phi\}), A^+\{\phi^+\}), \mathsf{Id}_{A\{\phi\}}) \} \\ 
= & \ R^{\textsf{Id}}_{A\{\phi\}, B\{\phi^{+++}\}}(\tau \bullet ((\phi \bullet \mathsf{p}(A\{ \phi \})) \& \mathsf{v}_{A\{\phi\}})) \\
= & \ R^{\textsf{Id}}_{A\{\phi\}, B\{\phi^{+++}\}}(\tau\{\phi^+\}).
\end{align*}

\end{itemize}

\end{proof}

\subsubsection{Game-theoretic Universes}

Next, we equip the CwF $\mathcal{PGD}$ with the game-theoretic universes. For the general, categorical definition of \emph{semantic universes}, see \cite{yamada2016game}.
\begin{proposition}[$\mathcal{PGD}$ supports universes]
The CwF $\mathcal{PGD}$ of predicative gamoids supports universes.
\end{proposition}
\begin{proof}
Let $\Gamma \in \mathcal{PGD}$ be any context gamoid.
\begin{itemize}

\item {\bfseries \sffamily U-Form.} For each natural number $n \in \mathbb{N}$, the dependent gamoid $\mathcal{U}_n \in \mathsf{Ty}(\Gamma)$ is the trivial one such that $\gamma \mapsto \mathcal{C}(\mathcal{U}_n)$, $p \mapsto \textsf{id}_{\mathcal{C}(\mathcal{U}_n)}$ for all $\gamma : \Gamma$, $p : \gamma =_\Gamma \! \gamma'$, where $\mathcal{U}_n$ is the $n^\text{th}$ universe game. 

\item {\bfseries \sffamily U-Intro, Elim, and Comp.} For any dependent gamoid $G \in \mathsf{Ty}(\Gamma)$, we have the strategy $G\gamma : \mathcal{U}_n$ for some $n \in \mathbb{N}$, for each $\gamma \in \Gamma$, and the functor $Gp : G\gamma \to G\gamma'$ for each equality $p : \gamma =_\Gamma \! \gamma'$, so it clearly induces the object
\begin{equation*}
\textstyle G \in \widehat{\prod}(\Gamma, \mathcal{U}_n).
\end{equation*}
Similarly, the strategy $\mathcal{U}_n : \mathcal{U}_{n+1}$ for each $n \in \mathbb{N}$ induces the object
\begin{equation*}
\textstyle \mathcal{U}_n \in \widehat{\prod}(\Gamma, \mathcal{U}_{n+1}).
\end{equation*}

\item {\bfseries \sffamily U-Cumul.} If $G \in \widehat{\prod}(\Gamma, \mathcal{U}_n)$, then clearly $G \in \widehat{\prod}(\Gamma, \mathcal{U}_{n+1})$ by the definition of the universe games.
\end{itemize}
\end{proof}

\section{Intensionality}
\label{Intensionality}
We now investigate \emph{how intensional the model of ITT in $\mathcal{PGD}$ is} through some of the rules in the type theory. We write $a =_A \! a'$ or just $a = a'$ for the Id-type of the terms $a, a' : A$ and $\vdash a \equiv a' : A$ for the judgemental equality in ITT.

\subsection{Equality Reflection}
The principle of \emph{equality reflection (EqRefl)}, which states that if two terms are propositionally equal, then they are judgementally equal too, is the difference between ITT and ETT: Roughly, ETT is ``ITT plus EqRefl''. 

It is straightforward to see that the model in $\mathcal{PGD}$ refutes EqRefl, as two computationally equal strategies are not necessarily strictly equal. Hence, it is a model of ITT, not ETT.

\subsection{Function Extensionality} 
Next, we consider the axiom of \emph{function extensionality (FunExt)} which states that: For any type $A$, dependent type $B : A \to \mathcal{U}$, and terms $f, g : \prod_{x:A}B(x)$, we can inhabit the type
\begin{equation*}
\textstyle \prod_{x:A} f(x) = g(x) \to f = g.
\end{equation*}

In the same way as \cite{hofmann1998groupoid} did, the model in $\mathcal{PGD}$ admits this principle. To see this explicitly, let $B : A \to \mathcal{PGD}$ be a dependent gamoid, and $\phi, \psi \in \widehat{\prod} (A, B)$ dependent functions. Assume that there is an object $\tau \in \widehat{\prod} (A, \mathsf{Id}_{B}\{ \phi \& \psi \})$, where note that $\mathsf{Id}_{B}\{ \phi \& \psi \} : A \to \mathcal{PGD}$ is a dependent gamoid. We then have an identification $\tau_p : \mathsf{Id}_{B}\{ \phi \& \psi \}(p) \bullet (\tau \bullet \sigma) =_{\phi \bullet \sigma' =_{B\sigma'} \psi \bullet \sigma'} \! \tau \bullet \sigma'$, i.e., 
\begin{equation*}
\tau_p : (\psi_p) \odot (\tau \bullet \sigma) \odot (\phi_p)^{\star} =_{\phi \bullet \sigma' =_{B\sigma'} \psi \bullet \sigma'} \! \tau \bullet \sigma'
\end{equation*}
for each $p : \sigma =_A \! \sigma'$. But it implies the naturality condition for the family $(\tau \bullet \sigma)_{\sigma : A}$ because there is no non-trivial identification in Id-gamoids. Therefore, we may take $\& \{ \tau \bullet \sigma \ \! | \ \! \sigma : A \}$ as an identification between $\phi$ and $\psi$.

\subsection{Uniqueness of Identity Proofs}
Next, we investigate the principle of \emph{uniqueness of identity proofs (UIP)} which states that: For any type $A$, the following type can be inhabited
\begin{equation*}
\textstyle \prod_{a_1, a_2 : A} \prod_{p, q : a_1 = a_2} p = q
\end{equation*}
Remarkably, the model in $\mathcal{PGD}$ refutes UIP, which is the main improvement in comparison with the model in the CwF $\mathcal{IPG}$ in Part I (\cite{yamada2016game}). Consider the \emph{boolean gamoid} $\mathsf{B}$ whose plays are prefixes of the sequences $q_{\mathsf{tt}} . \mathsf{tt}, q_{\mathsf{ff}} . \mathsf{ff}$ with all isomorphism strategies as identifications (i.e., it is a canonical gamoid). Let us write $\bullet : \mathsf{B}$ for the unique total strategy. Then explicitly, the identifications are the copy-cat strategy $\textsf{cp}_{\mathsf{B}}$ and the ``reversing" strategy $\textsf{rv}_{\mathsf{B}}$. We then have $\textsf{cp}_{\mathsf{B}} \neq_{\bullet =_{\mathsf{B}} \bullet} \textsf{rv}_{\mathsf{B}}$ because the identifications in $\bullet =_{\mathsf{B}} \bullet$ are only the trivial ones.

\begin{remark}
This argument is essentially the same as how the groupoid model in \cite{hofmann1998groupoid} refutes UIP.
\end{remark}

\subsection{Criteria of Intensionality}
There are Streicher's three \emph{Criteria of Intensionality}:
\begin{itemize}
\item {\bf I.} $A : \mathcal{U}, x, y : A, z : x=_A \! y \not \vdash x \equiv y : A$
\item {\bf II.} $A : \mathcal{U}, B : A \to \mathcal{U}, x, y : A, z : x=_A \! y \not \vdash B(x) \equiv B(y) : \mathcal{U}$
\item {\bf III.} If $\vdash p : t =_A \! t'$, then $\vdash t \equiv t' : A$
\end{itemize} 
It is straightforward to see that the model in $\mathcal{PGD}$ validates the criteria I and II but refutes the criterion III. Note that HoTT has the criteria I and II but not III.

\subsection{Univalence}
We finally analyze the \emph{univalence axiom (UA)}, the heart of HoTT, which states that
\begin{equation*}
(A =_\mathcal{U} \! B) \simeq (A \simeq B)
\end{equation*}
for all types $A$ and $B$ (for the definition of $\simeq$, see \cite{voevodsky2013homotopy}). 
It is then easy to see that this axiom holds for the model in $\mathcal{PGD}$ because identifications in the universe gamoids are isomorphism strategies. Our definition of morphisms in a gamoid (they are not qcc strategies but isomorphism strategies) and interpretation of universes as canonical gamoids was mainly to establish this result. 

However, note that we only have the trivial identifications between identifications. Thus, it is a future work to interpret UA as well as the infinite hierarchy of Id-types.

\section{Conclusion}
\label{Conclusion}
In the present paper, we defined a new game-theoretic interpretation of intensional type theory with $\prod$-, $\sum$-, and Id-types as well as universes. It can be seen essentially as a concrete instance of the groupoid model developed by Hofmann and Streicher \cite{hofmann1998groupoid}.

Our model refutes UIP and admits UA as well as FunExt, though it does not interpret non-trivial higher propositional equalities. Thus, in Part III, we shall generalize predicative gamoids to be an instance of $\omega$-groupoids to interpret the hierarchy of Id-types as the models in \cite{warren2011strict, van2011types, lumsdaine2009weak} did. As another future work, we shall address the problem of definability and full abstraction.

Finally, note that, comparing with abstract, categorical models of the type theory such as the ($\omega$-)
groupoid model, our game-theoretic model is very concrete; and in contrast
with homotopy-theoretic models, our model directly represents computations or algorithms of the type theory in an intuitive manner. In this sense, we believe that our model is not merely a tool to analyze the syntax but a mathematical formulation of the philosophy and concepts on the type theory.

\pagebreak

\bibliographystyle{alpha}
\bibliography{GamesAndStrategies,HoTT,Recursion}

\end{document}